%% file: main.tex
\documentclass[11pt]{article}

\usepackage{fullpage}
\usepackage{svg}
\usepackage{amsmath}
\usepackage{amssymb, amsthm, amsfonts}
\usepackage{bm,mathtools}
\usepackage[inline,shortlabels]{enumitem}
\usepackage{csquotes}
\usepackage[colorinlistoftodos,prependcaption,textsize=tiny]{todonotes}
\usepackage{hhline}
\usepackage{multicol}
\usepackage{framed}
\usepackage[framemethod=tikz]{mdframed}
\usepackage{graphicx}
\usepackage{color}

\usepackage{multirow}
\usepackage{algorithm, subcaption}
\PassOptionsToPackage{boxed,section}{algorithm}
\usepackage[noend]{algpseudocode}
\usepackage{float}
\usepackage{calrsfs}
\DeclareMathAlphabet{\oldmathcal}{OMS}{zplm}{m}{n}

\usepackage{thmtools}
\usepackage{thm-restate}
\usepackage{adjustbox}
\usepackage{xspace}
\definecolor{darkgreen}{rgb}{0,0.5,0}
\definecolor{darkblue}{rgb}{0,0,0.8}
\usepackage{hyperref}
\hypersetup{
    unicode=false,          %
    colorlinks=true,        %
    linkcolor=darkblue,          %
    citecolor=darkgreen,        %
    filecolor=magenta,      %
    urlcolor= black           %
}
\RequirePackage[]{silence}
\usepackage[nameinlink,capitalize]{cleveref}

\newtheorem{theorem}{Theorem}[section]
\newtheorem{fact}[theorem]{Fact}

\newtheorem{lemma}[theorem]{Lemma}

\newtheorem{corollary}[theorem]{Corollary}

\newtheorem{property}[theorem]{Property}

\theoremstyle{definition}
\newtheorem{definition}[theorem]{Definition}

\mathchardef\mhyphen="2D

\newcommand{\vol}{\operatorname{\text{{\rm vol}}}}

\newcommand{\ignore}[1]{}

\algnewcommand\algorithmicswitch{\textbf{switch}}
\algnewcommand\algorithmiccase{\textbf{case}}

\algdef{SE}[SWITCH]{Switch}{EndSwitch}[1]{\algorithmicswitch\ #1\ \algorithmicdo}{\algorithmicend\ \algorithmicswitch}%
\algdef{SE}[CASE]{Case}{EndCase}[1]{\algorithmiccase\ #1}{\algorithmicend\ \algorithmiccase}%
\algtext*{EndSwitch}%
\algtext*{EndCase}%

\newcommand{\congest}{\ensuremath{\mathsf{CONGEST}}\xspace}

\newcommand{\pram}{\ensuremath{\mathsf{PRAM}}\xspace}
\newcommand{\local}{$\mathsf{LOCAL}$\xspace}

\newcommand{\poly}{\operatorname{\text{{\rm poly}}}}

\newcommand{\polylog}{\operatorname{\text{{\rm polylog}}}}

\newcommand{\ID}{\operatorname{ID}}
\newcommand{\dist}{\operatorname{dist}}

\DeclareMathOperator{\E}{\mathbb{E}}

\newcommand{\T}{\oldmathcal{T}}

\algnewcommand{\LineComment}[1]{\Statex \(\triangleright\) #1}

\title{Deterministic Expander Routing: Faster and More Versatile}
\author{
Yi-Jun Chang\thanks{Supported by the NUS Presidential Young Professorship startup grant.}\\
National University of Singapore
\and 
Shang-En Huang\thanks{Supported by NSF CCF-2008422.}\\
Boston College
\and 
Hsin-Hao Su\footnotemark[2]\\
Boston College
}

\date{}

\begin{document}

\maketitle
\thispagestyle{empty}
\begin{abstract}
We consider the expander routing problem formulated by Ghaffari, Kuhn, and Su (PODC 2017), where the goal is to route all the tokens to their destinations given that each vertex is the source and the destination of at most $\deg(v)$ tokens. They developed  {\it randomized algorithms} that solve this problem in $\poly(\phi^{-1}) \cdot 2^{O(\sqrt{\log n \log \log n})}$ rounds in the \congest model, where $\phi$ is the conductance of the graph. In addition, as noted by Chang, Pettie, Saranurak, and Zhang (JACM 2021), it is possible to obtain a preprocessing/query tradeoff so that the routing queries can be answered faster at the cost of more preprocessing time. The efficiency and flexibility of the processing/query tradeoff of expander routing have led to many other distributed algorithms in the \congest model, such as subpolynomial-round minimum spanning tree algorithms in expander graphs and near-optimal algorithms for $k$-clique enumeration in general graphs. 

As the routing algorithm of Ghaffari, Kuhn, and Su and the subsequent improved algorithm by Ghaffari and Li (DISC 2018) are both randomized, all the resulting applications are also randomized. Recently, Chang and Saranurak (FOCS 2020) gave a deterministic algorithm that solves an expander routing instance in $2^{O(\log^{2/3} n \cdot \log^{1/3} \log n)}$ rounds. The deterministic algorithm is less efficient and does not allow preprocessing/query tradeoffs, which precludes the de-randomization of algorithms that require this feature, such as 
the aforementioned $k$-clique enumeration algorithm in general graphs.

The main contribution of our work is a new deterministic expander routing algorithm that not only matches the randomized bound of Ghaffari, Kuhn, and Su but also allows preprocessing/query tradeoffs. Our algorithm solves a single instance of routing query in $2^{{O}(\sqrt{\log n \cdot  \log \log n})}$ rounds. For instance, this allows us to compute an MST in an expander graph in the same round complexity deterministically, improving the previous state-of-the-art $2^{O(\log^{2/3} n \cdot \log^{1/3} \log n)}$.
Our algorithm achieves the following preprocessing and query tradeoffs: For $0 < \epsilon < 1$, we can answer every routing query in $\log^{O(1/\epsilon)} n$ rounds at the cost of a $(n^{O(\epsilon)} + \log^{O(1/\epsilon)} n)$-round preprocessing procedure. Combining this with the approach of Censor-Hillel, Leitersdorf, and Vulakh (PODC 2022), we obtain a near-optimal $\tilde{O}(n^{1-2/k})$-round deterministic algorithm for $k$-clique enumeration in general graphs, improving the previous state-of-the-art  $n^{1-2/k+o(1)}$.

As a side result of independent interest, we demonstrate the \emph{equivalence} between expander routing and \emph{sorting} in the sense that they are reducible to each other up to a polylogarithmic factor in round complexities in the \congest model.

\end{abstract}

\newpage 
{ \normalsize	
\tableofcontents
}
\thispagestyle{empty} %
\newpage
\newpage

\setcounter{page}{1}
\input{introduction}
\input{hierarchical_decomposition}

\input{recursive_routing}

\bibliographystyle{alpha}
\bibliography{references}
\input{appendix}

\input{Task123}
\end{document}

%% file: introduction.tex
\section{Introduction}

The \congest model is a prominent model that captures both the locality and the bandwidth in the study of distributed graph algorithms. In this model, the underlying network is a graph $G=(V,E)$, where we let $n=|V|$, $m=|E|$, and $\Delta =$ the maximum degree of $G$. Every vertex $v$ hosts a processor with an unique $\ID \in [1,\poly(n)]$. The computation proceeds in synchronized rounds. In each round, each vertex sends a distinct message of $O(\log n)$ bits to each of its neighbors, receives messages from its neighbors, and performs local computations. The complexity of an algorithm is measured as the number of rounds. 

In this work, we focus on networks with high {\it conductance}. 
Throughout the paper, we say that a graph is a \emph{$\phi$-expander} if its conductance is at least $\phi$, and we informally say that a graph is an \emph{expander} if it has high conductance. Depending on the context, the conductance of an expander can be $\Omega(1)$, $\log^{-O(1)} n$, or $n^{-o(1)}$.

We consider the following routing problem in a $\phi$-expander $G$ in the \congest model. Suppose that each vertex $v$ is the source and the destination of at most $\deg(v)$ tokens. The goal is to route all the tokens to their destinations. Ghaffari, Kuhn, and Su~\cite{GKS17} developed an algorithm that routes the tokens in $\poly(\phi^{-1}) \cdot 2^{O(\sqrt{\log n \log \log n})}$ rounds. By using such a primitive, 
a minimum spanning tree (MST) can be computed in $\poly(\phi^{-1}) \cdot 2^{O(\sqrt{\log n \log \log n})}$ rounds in the \congest model, beating the $\Omega(\sqrt{n} / \log n)$ lower bound of \cite{PR00, SHK12} in general graphs. Later, the $2^{O(\sqrt{\log n \cdot \log \log n})}$ term in the running time has been improved to ${2^{O(\sqrt{\log n})}}$ later by Ghaffari and Li~\cite{GL18}.

Chang, Pettie, Saranurak, and Zhang~\cite{CPZ21} leveraged the expander routing algorithms to general graphs by developing distributed algorithms for expander decomposition. They showed the method can be used to obtain efficient algorithms for a series of problems. In particular, they obtained  \congest algorithms for triangle counting, detection, and enumeration whose running times match the triangle enumeration lower bound of \cite{IzumiL17} up to $\polylog(n)$ factors. The approach is to decompose the input graph into disjoint expanders, where only a small number of edges are crossing between different expanders. Within each expander, the ease of routing provided by these algorithms allows one to solve the problem efficiently. They also noted that the algorithms of \cite{GKS17} can be tweaked to have preprocessing/query tradeoffs and used this perk in obtaining the above optimal-round algorithms. In particular, if one spends $O(n^{\epsilon})$ time doing the preprocessing then each subsequent routing instance can be answered in $O(\log^{O(1/\epsilon)} n)$ time. This is particularly useful for algorithms that need a polynomial number of queries, as each query can be answered in polylogarithmic rounds if we spend a small-polynomial time for preprocessing.

One major issue left by \cite{GKS17, GL18} was that their routing algorithms are randomized. As a result, all the resulting applications are randomized. In \cite{ChangS20}, they made progress by giving an deterministic algorithm that solves a routing instance in $\poly(\phi^{-1}) \cdot 2^{O(\log^{2/3} n \cdot \log^{1/3}\log n )}$, which is suboptimal compared to the randomized bound of $\poly(\phi^{-1}) \cdot 2^{O(\sqrt{\log n\cdot \log\log n})}$. More importantly, it did not achieve processing/query tradeoffs as in \cite{GKS17}. Therefore, for many applications of the deterministic expander routing, such as the aforementioned results for triangle detection and triangle enumeration, it induces an additional factor of $2^{O(\log^{2/3} n \cdot \log^{1/3}\log n )}$, leaving a substantial gap between randomized and deterministic algorithms.

\subsection{Our Contribution}
The main contribution of our paper is a deterministic expander routing algorithm that matches the randomized bound of \cite{GKS17} with preprocessing/routing tradeoffs.

\begin{theorem}\label{thm:main} Given a graph $G=(V,E)$ be a $\phi$-expander. Let $\epsilon > 0$ be a constant. There exists an algorithm that preprocesses the graph in $n^{O(\epsilon)} + \poly(\phi^{-1})\cdot (\log n)^{O(1/\epsilon)}$ time such that each subsequent routing instance can be solved in $\poly(\phi^{-1})\cdot (\log n)^{O(1/\epsilon)}$ rounds. \end{theorem}

Here we see that a single routing instance can be solved in time similar to the bounds obtained by \cite{GKS17} by setting $\epsilon = \sqrt{\log \log n / \log n}$ in \cref{thm:main}.

\begin{corollary}\label{cor:single-routing} A single expander routing instance can be solved in $\poly(\phi^{-1}) \cdot 2^{{O}(\sqrt{\log n \cdot \log \log n})}$ rounds deterministically without preprocessing. \end{corollary}

\cref{cor:single-routing} is an improvement over the previous deterministic expander routing algorithm of~\cite{ChangS20}, which costs $\poly(\phi^{-1}) \cdot 2^{O({\log^{2/3} n} \log^{1/3} \log n)}$ rounds.

Expander routing is extremely useful as a fundamental \emph{communication primitive} in designing distributed algorithms in expander graphs. Expander routing has been used to design efficient MST and minimum cut algorithms~\cite{GKS17}, efficient subgraph finding algorithms~\cite{CPZ21}, and efficient algorithms for sorting, top-$k$ frequent elements, and various data summarization tasks~\cite{su2019distributed} in expander graphs. 
Expander routing allows us to transform a large class of work-efficient \pram algorithms into \congest algorithms with small overhead~\cite{GL18}. Expander routing has also been utilized in a smooth analysis for distributed MST~\cite{chatterjee2020distributed} and to design sparsity-aware algorithms for various shortest path computation tasks in the \congest model~\cite{censor2021sparsity}.

Our improved expander routing algorithm immediately leads to improved deterministic upper bounds for all of the aforementioned applications. In particular, our result implies that an MST of an $\phi$-expander can be computed in $\poly(\phi^{-1}) \cdot 2^{O(\sqrt{\log n \log \log n})}$ rounds deterministically, improving upon the previous deterministic bound $\poly(\phi^{-1}) \cdot 2^{O({\log^{2/3} n} \log^{1/3} \log n)}$~\cite{ChangS20} and nearly matching the current randomized bound $\poly(\phi^{-1}) \cdot 2^{O(\sqrt{\log n})}$~\cite{GKS17,GL18}.

\begin{corollary} An MST of an $\phi$-expander can be computed in  $\poly(\phi^{-1}) \cdot 2^{{O}(\sqrt{\log n \log \log n})}$ rounds deterministically. \end{corollary}
\begin{proof}
Similar to the randomized MST algorithm in~\cite{GKS17}, it was shown in~\cite{ChangS20} that an MST can be constructed using polylogarithmic deterministic rounds and invocations of expander routing. Therefore, an MST of an $\phi$-expander can be computed in  $\poly(\phi^{-1}) \cdot 2^{{O}(\sqrt{\log n \log \log n})}$ rounds deterministically by implementing the MST algorithm using the expander routing algorithm of \cref{cor:single-routing}.
\end{proof}

Expander routing is also useful in designing distributed algorithms in \emph{general graphs} indirectly via the use of \emph{expander decompositions}. An $(\epsilon, \phi)$ expander decomposition of a graph removes at most $\epsilon$ fraction of the edges in such a way that each remaining connected component induces a $\phi$-expander. 
In the \congest model, this decomposition is commonly applied in a divide-and-conquer approach, where efficient expander routing algorithms are employed to solve subproblems within $\phi$-expanders. This approach has been particularly successful in the area of distributed subgraph finding~\cite{CPZ21,ChangS20,censor2020distributed,censor2021tight,censor2022deterministic,censor2022quantum,eden2022sublinear,izumi2020quantum,le2021lower}. A different use of expander decompositions and routing is to establish barrier for proving lower bounds in \congest~\cite{eden2022sublinear}.

Again, our improved deterministic expander routing algorithm leads to improved bounds for the aforementioned applications. In particular, we obtain a near-optimal $\tilde{O}(n^{1-2/k})$-round deterministic algorithm for $k$-clique enumeration in general graphs, improving the previous deterministic upper bound $n^{1-2/k+o(1)}$~\cite{censor2022deterministic}.

\begin{corollary}\label{cor:clique_listing} There is a deterministic algorithm that list all $k$-cliques in  $\tilde{O}(n^{1-2/k})$ rounds deterministically. \end{corollary}
\begin{proof}
By slightly modifying the algorithm of~\cite{censor2022deterministic}, we know that all $k$-cliques can be listed using $\tilde{O}(n^{1-2/k})$ deterministic rounds and invocations of expander routing on $\phi$-expanders with $\phi = 1/\polylog(n)$. The modification needed is to alter the parameters for the deterministic $(\epsilon, \phi)$ expander decomposition in~\cite[Theorem 5]{censor2022deterministic}. Here we want to make $\phi = 1/\polylog(n)$. 

As discussed in~\cite{ChangS20}, the deterministic $(\epsilon, \phi)$ expander decomposition algorithm admits the following tradeoff: for any $1 \geq \gamma \geq \sqrt{\log \log n / \log n}$, there is a deterministic expander decomposition algorithm with round complexity $\epsilon^{-O(1)}  \cdot n^{O(\gamma)}$ with parameter $\phi = \epsilon^{O(1)}  \log^{-O(1/\gamma)}n$. In the $k$-clique enumeration algorithm of~\cite{censor2022deterministic}, the parameter $\epsilon$ is set to be some constant.
By selecting $\gamma$ to be a sufficiently large constant, we can ensure that $\phi = \epsilon^{O(1)}  \log^{-O(1/\gamma)}n = 1/\polylog(n)$ and the round complexity $\epsilon^{-O(1)}  \cdot n^{O(\gamma)}$ for constructing the decomposition is upper bounded by $\tilde{O}(n^{1-2/k})$.

If we implement the $k$-clique enumeration algorithm with the $\poly(\phi^{-1}) \cdot 2^{O(\sqrt{\log n \log \log n})}$-round deterministic expander routing algorithm of~\cite{ChangS20}, then the overall round complexity for  $k$-clique enumeration is $\tilde{O}(n^{1-2/k}) \cdot 2^{O(\sqrt{\log n \log \log n})} = n^{1-2/k+o(1)}$.
To improve the upper bound to $\tilde{O}(n^{1-2/k})$, we use our new deterministic expander routing algorithm. Specifically, by selecting $\epsilon$ to be a sufficiently small constant in \cref{thm:main}, we can ensure that each routing instance can be solved in $\poly(\phi^{-1})\cdot (\log n)^{O(1/\epsilon)} = \polylog(n)$ rounds and the cost  $n^{O(\epsilon)} + \poly(\phi^{-1})\cdot (\log n)^{O(1/\epsilon)}$ of the preprocessing step is upper bounded by $\tilde{O}(n^{1-2/k})$.
\end{proof}

Our algorithm is optimal up to a polylogarithmic factor, as the upper bound $\tilde{O}(n^{1-2/k})$ for $k$-clique enumeration in \cref{cor:clique_listing} matches the $\tilde{\Omega}(n^{1-2/k})$ lower bound~\cite{fischer2018possibilities,IzumiL17}. Previously, such an upper bound was only known to be achievable in the randomized setting~\cite{censor2021tight}. Moreover, for $k = 4$, our algorithm is tight even for the easier $k$-clique \emph{detection} problem, due to the $\tilde{\Omega}(\sqrt{n})$ $4$-clique detection lower bound of~\cite{czumaj2020detecting}.

\cref{thm:main,cor:clique_listing} resolve an open question of Censor-Hillel,\footnote{Open Problem 2.2 of \url{https://arxiv.org/abs/2203.06597v3}.} which asks whether the cost of each instance of expander routing in the triangle enumeration algorithm can be made both \emph{deterministic} and has a \emph{polylogarithmic} round complexity. \Cref{cor:clique_listing} yields
a deterministic triangle enumeration algorithm that is optimal up to a polylogarithmic factor.

\subsection{Previous Results and Key Challenges}
For ease of discussion, in this section, we assume that our input graph has an $O(1)$ maximum degree and is an expander with constant conductance.

\paragraph{Randomized Approach}  We first summarize at a high level the general idea of \cite{GKS17} and explain the difficulty of de-randomization. Roughly speaking, the general idea is to partition the current base graph $X$ into $k=n^{\epsilon}$ parts $X_1, \ldots, X_k$ with roughly equal sizes. For each part $X_i$, by using random walk techniques, they embed a virtual Erd\H{o}s--Renyi graph $G(|X_i|,p)$ onto it for $p=O(\log n/  |X_i|)$, where all the virtual edges correspond to a set of paths $\mathcal{P}$ with $\polylog(n)$ {\it congestion} and {\it dilation} in $X$, where the congestion $c$ is defined to be $c=\max_{e} |\{P \ni e \mid P \in \mathcal{P} \} |$ and the dilation $d$ is defined to be $d= \max_{P \in \mathcal{P}} |P|$. The quantity $c+d$ is known as the quality of $\mathcal{P}$ or the quality of the embedding, as one round of communication in the virtual graph can be simulated within $O(cd)$ rounds in the base graph deterministically, and $\tilde{O}(c+d)$ rounds with randomization  \cite{LMR94, Ghaffari15}. As Erd\H{o}s--Renyi graphs are good expanders, they may recurse on each $X_i$ by viewing the base graph as the virtual graph $G(|X_i|,p))$ to further partition $X_i$ into $k$ parts and embed a $G(n,p)$ on each of them.  The hierarchy goes on for $O(1/\epsilon)$ levels. Since each level only incurs a $\polylog(n)$ blow up on the congestion and dilation. A set of paths of subgraphs in any level with quality $c+d$ corresponds to a set of paths in the original graph of quality $(c+d)\cdot \log^{O(1/\epsilon)} n$.  With such a hierarchy embedding structure, they showed a routing instance can be routed using paths that consist of edges in the virtual graphs across different levels with quality $\log^{O(1/\epsilon)} n$, which translates to paths of quality $(\log^{O(1/\epsilon)} n)^2 = \log^{O(1/\epsilon)} n$ in the original graph.

Now we examine the deterministic routing algorithm of \cite{ChangS20} and address the reasons why it did not obtain the randomized bound and the preprocessing/query tradeoffs.%

\paragraph{Challenge I -- Speed} 
At a high level, the deterministic routing algorithm of \cite{ChangS20} still follows the same recursive framework used in the randomized algorithm of~\cite{GKS17}. While a low-congestion and low-dilation simultaneous embedding of virtual expanders into $X_1, \ldots, X_k$ can be obtained easily by random walks, obtaining such a simultaneous embedding of expanders is much more difficult in the deterministic setting. In \cite{ChangS20}, low-congestion and low-dilation simultaneous embedding of virtual expanders is computed recursively using an approach similar to that of~\cite{chuzhoy2020deterministic} based on the \emph{cut-matching game} of~\cite{KKOV07}. 

We give a brief and informal introduction to how the cut-matching game works. The cut-matching game is a procedure that returns a balanced sparse cut or a low-congestion and low-dilation embedding of a virtual expander. The algorithm works by iteratively finding a sparse cut of the virtual graph and then finding a low-congestion and low-dilation embedding of a large matching between the two parts of the cut. If we cannot obtain a large matching at some stage of the algorithm, then a balanced sparse cut can be obtained. Otherwise, the virtual graph is guaranteed to be an expander. In \cite{ChangS20,chuzhoy2020deterministic}, the implementation of the sparse cut algorithm in the cut-matching game is done recursively with a recursive structure similar to that of~\cite{GKS17} where recursion is applied to multiple smaller instances. 

Due to the recursive nature of the approach discussed above, the deterministic simultaneous embedding of virtual expanders in~\cite{ChangS20} has a much worse guarantee compared to the randomized approach of~\cite{GKS17}: Specifically, within $(n^{O(\epsilon)}+\log^{O(\log (1/\epsilon))} n)$ rounds, the expanders obtained have conductance of $1/(\log^{O(1/\epsilon)} n)$.
As discussed earlier, to build the hierarchical structure needed to solve the routing problem,  one has to repeat the process of simultaneous embedding of virtual expanders recursively, and the depth of recursion is $O(1/\epsilon)$. Since each level incurs a blow-up of $\log^{O(1/\epsilon)} n$ factor on the routing quality, from the bottom to the top, it introduces a $\log^{O(1/\epsilon^2)} n$ blow-up in total, as opposed to $\log^{O(1/\epsilon)} n$ in the randomized construction of~\cite{GKS17}.  By balancing the terms $(n^{O(\epsilon)}+\log^{O(\log (1/\epsilon))} n)$ and $\log^{O(1/\epsilon^2)} n$, it turns out setting $\epsilon = (\log \log n/\log n)^{1/3}$ yields the best possible bound of $2^{O(\log^{2/3}\cdot \log^{1/3}\log n)}$, which is sub-optimal compared to the randomized algorithms of~\cite{GKS17}.

\paragraph{Challenge II -- Preprocessing/Query Tradeoffs} In the randomized routing algorithm of~\cite{GKS17}, it is possible to obtain a preprocessing/query tradeoff, where the preprocessing phase builds the hierarchy of expander embeddings in $(n^{O(\epsilon)}+\log^{O(\log (1/\epsilon))} n)$ rounds. Each routing query in the query phase can be done in $\log^{O(\log (1/\epsilon))} n$ rounds. Very different from the randomized approach, the deterministic routing algorithm of~\cite{ChangS20} still requires  $(n^{O(\epsilon)}+\log^{O(\log (1/\epsilon^2))} n)$ rounds for every routing query, so a tradeoff between preprocessing and query cannot be achieved.

We briefly explain why the disparity occurs. In the randomized setting~\cite{GKS17}, the \emph{same} collection of routing paths constructed in the preprocessing step can be reused for \emph{all} subsequent routing requests that are \emph{oblivious} to the randomness used in the preprocessing step. Such an oblivious assumption can be made without loss of generality by first using random walks to redistribute the messages to be routed. In the deterministic setting~\cite{ChangS20}, the paths for routing the messages are recomputed from scratch for each routing request, as we explain below.

Suppose the current base graph is $X$. Let $X_1 \ldots X_k$ be the children of $X$ in the hierarchy. We classify the tokens needed to be routed $T_1 \ldots T_k$ based on their destinations, where $T_i$ is the set of tokens whose destinations are in $X_i$. The routing task of the current level of recursion is to route all the tokens $T_i$ to $X_i$. Once such a task has been achieved, we can just recurse in each $X_i$.   
The deterministic algorithm of \cite{ChangS20} resolves this task by iterating over all the $O(k^2)$ $X_i$-$X_j$ pairs sequentially. For each $X_i$-$X_j$ pair, they find a set of paths to send the tokens $T_j$ from $X_i$ to $X_j$ with quality $\poly(k) \cdot 2^{O(\sqrt{\log n})}$ by adapting the \emph{maximal paths} algorithm in \cite{GPV93}, which were originally used to compute matching and DFS in \pram. As a result, there is a $\poly(k) = n^{O(\epsilon)}$ dependency on the query complexity, which is not needed in the randomized algorithm of~\cite{GKS17}.

\subsection{Our Approach}\label{sec:our-approach}
We describe how we overcome the above two challenges as follows.  First, to get the bound that matches the randomized algorithm of \cite{GKS17}, we do a one-shot hierarchical decomposition.

\paragraph{One-Shot Hierarchical Decomposition} Instead of applying the deterministic simultaneous expander embedding framework~\cite{ChangS20} as a black box and recursing on each embedded expander to build the embedding hierarchy, we observe that for the algorithm of \cite{ChangS20} to return such an embedding of expanders in one level, the algorithm already builds some kind of a hierarchy of expander embedding during the recursive construction. Therefore, a natural idea for improving the deterministic routing algorithm of~\cite{ChangS20} is to run the simultaneous expander embedding algorithm {\bf only once} in the base level and use the hierarchical decomposition constructed in the algorithm to solve the routing problem in a way similar to that of~\cite{GKS17,ChangS20}. To realize this idea, we need to overcome some technical difficulties. In particular, here each level in the hierarchy not only introduces a loss in the conductance guarantee but also a loss in the number of vertices covered by the expander embedding, as the hierarchical decomposition only embeds expanders on a constant fraction of vertices in each level. One observation of why such an approach is still plausible is that the depth of the hierarchy is $O(1/\epsilon)$, so the expanders at the bottom level consist of $1/2^{O(1/\epsilon)}$ fraction of the vertices. Therefore, it might be possible to find delegates in those bottom-level expanders, which we will refer to as the {\it best nodes},  for every vertex in the original graph in such a way that each best node represents at most $2^{O(1/\epsilon)}$ vertices.  This would incur at most $2^{O(1/\epsilon)}$ blow up on the congestion. Moreover, the edges in the virtual expanders in each level of the hierarchy correspond to paths of quality at most $\polylog(n)$ in the parent level. The total blow up on the quality is at most $(\log^{O(1/\epsilon)} n)^2 = \log^{O(1/\epsilon)} n$. This is in contrast with the algorithm of \cite{ChangS20}, which has a blow-up of $\log^{O(1/\epsilon^2)} n$. 

We define additional tasks and reduce the original problem to these tasks to implement the delegation idea. However, for the ease of illustration in the introduction, let us assume for now a base graph $X$ is a partition into $X_1 \ldots X_k$, where an expander can be embedded into each $X_i$. Also, the hierarchy has been constructed recursively on the expander of each $X_i$. 

\paragraph{A Randomized, Meeting in the Middle Approach} Second, to achieve a preprocessing/query tradeoff, given base graph $X$, we need a routing algorithm that has no polynomial dependencies on $k$ that routes the tokens to the corresponding parts. We first describe a randomized version of our approach and explain how to de-randomize it: Perform lazy random walks simultaneously for all the tokens together until they mix. For tokens destined to $X_i$ (call these tokens $T_i$), they are now roughly equally distributed across different parts. Suppose that we call such a configuration the {\it dispersed configuration} and the desired configuration the {\it final configuration}. To route from the dispersed configuration to the final configuration, we start with the final configuration, transform it into the dispersed configuration by the same method, and reverse the paths. The only problem left now is that the two dispersed configurations can be different, and we still need to match up $T_i$ tokens with $T'_i$ tokens for each $i$ inside each part $X_j$. Here, we can then embed a sorting network into each $X_j$ to sort the tokens so they are aligned to match up (see the Expander Sorting paragraph at the end of the section). 

\paragraph{De-randomization by Pre-embeddings of Shufflers} Now the only issue left is to remove the randomness needed in the process of routing tokens from any configuration to a dispersed configuration. The cut-matching game, introduced by \cite{KRV09}, is a potential deterministic way to achieve a similar effect of random walks. Roughly speaking, the goal of the game is to produce {\it a shuffler}, which consists of matchings of virtual edges $M^1, M^2, \ldots, M^\lambda$ such that the natural random walk on the sequence of matchings converges to a nearly uniform distribution from any initial distribution, where each $M^r$ corresponds to a set of paths of low congestion and dilation (i.e. if $(u,v) \in M^r$ then there is a $u$-$v$ path in the set). The natural random walk defined by  $M^1, M^2, \ldots, M^\lambda$ is a random walk such that for $r = 1,\ldots,\lambda$, if the current vertex $v$ is matched to $u$ then we move to $u$ (through its corresponding path) with probability $1/2$, and stay at $v$ with probability $1/2$. If $v$ is not matched, then it stays at $v$. 

Once we have such a shuffler, we can distribute the tokens deterministically according to the behavior of a lazy random walk. In particular, at each node $u$, consider if the number of $T_i$-tokens that are on $u$ is $x_i$. For $r =1 \ldots,\lambda$, if $u$ is matched to $v$ in $M^r$, we need to send $x_i/2$ $T_i$-token from $u$ to its mate $v$. Assuming the tokens are splittable (to be fractional). In the end, every node would hold a roughly equal amount of $T_i$ tokens due to the mixing property of the shuffler. This would lead to the dispersed configuration. 

\paragraph{Coarse-grained Shufflers}
However, the tokens are not splittable. To this end, instead of building a shuffler on $X$, we build a shuffler on $Y$, where $Y$ is a multi-graph obtained from $X$ by contracting each $X_i$. By doing such a coarse-grained shuffling, the rounding error due to the integrality of the tokens becomes negligible when $|X_i| \gg |Y|$. 

Yet, directly running the cut-matching game on $Y$ will lead to insufficient bandwidth for token distribution. If $X_j$ is matched $X_{j'}$ by the matching player, then we need to send $x_i/ 2$ $T_i$-tokens from $X_j$ to $X_{j'}$ in the simulation of lazy random walk, where $x_i$ is the number of $T_i$-tokens on $X_j$. Since each matched edge only corresponds to one path and it can be the case that $x_i = \omega(1)$, the bandwidth may not be enough.

To resolve this, we implement the cut player on $Y$ and the matching player on $X$ to ensure the matching player finds enough paths. This will lead the algorithm to produce a shuffler consisting of matchings of $X$ along their path embeddings of low congestion and dilation. The matchings of $X$ can be naturally translated to fractional matchings of $Y$ by normalization. We then simulate the token distribution on $Y$ according to these fractional matchings, using the path embeddings in $X$. %

\paragraph{Routing to Shuffler Portals} Once the shufflers are constructed, it will be ready to process queries of routing instances.  Recall that paths that correspond to a matching of the shuffler will be used to transport the tokens. The endpoint of such paths is known as {\it portals}. To route the tokens according to the fractional matchings, the main task is to send them to the corresponding portals so that they can follow the paths to the corresponding parts. For example, suppose there are $x_i$ $T_i$-tokens on $X_j$ for each $i$. If according to the fractional matching, we need to send $x_{i,j'}$ tokens to $X_{j'}$ then we need to route these $x_{i,j'}$ tokens to the portals in $X_i$. The routing tasks stemming from processing a fractional matching now become parallel instances of the routing task on each $X_i$.  The cut-matching games end in $O(\log n)$ iterations. So the problem recurses into $O(\log n)$ of parallel routing instances of the next level. To load balance the tokens over the portals, we again use the expander sorting technique to resolve it without dependency on $\poly(k)$. As a result, a query can be answered without dependency on $\poly(k)$.

\paragraph{Expander Sorting}
One particular subroutine---deterministic expander sorting---serves as a core tool in our routing algorithm. It has been used in, e.g., the aforementioned procedure for routing tokens to shuffler portals as well as other procedures such as re-writing token destinations and solving the problem within leaf components.

The goal of expander sorting is to re-distribute all tokens among the vertices such that, if we collect all the tokens from the vertex with the smallest ID to the vertex with the largest ID, these tokens' pre-defined keys are sorted in non-decreasing order. Su and Vu~\cite{su2019distributed} considered a slightly simpler version of the problem where each vertex holds a unique ID from $[1,n]$ and gave a randomized algorithm for it. Here, the IDs can range from $[1,\poly(n)]$. We gave deterministic algorithms for expander sorting along the way and developed several handy tools based on it. For example, gathering and propagating information with custom grouping keys.

\paragraph{The Equivalence Between Routing and Sorting} As a side result, we showed that expander routing and expander sorting tasks are actually \emph{equivalent} up to a polylogarithmic factor, in the sense that if there is a \congest algorithm $\oldmathcal{A}_{\textsf route}$ that solves the expander routing problem in $T_{\textsf{route}}(n, \phi, L)$ rounds, then an expander sorting instance can be solved within $O(\phi^{-1}\log n) + O(\log n)\cdot T_{\textsf{route}}(n, \phi, L)$ rounds.
Conversely, if there is a \congest algorithm $\oldmathcal{A}_{\textsf{sort}}$ that solves the expander sorting problem in $T_{\textsf{sort}}(n, \phi, L)$ rounds, then an expander routing instance can be solved within $O(1)\cdot T_{\textsf{sort}}(n, \phi, 2L)$ rounds.
We prove the equivalence in \Cref{sec:equivalence}.%

We believe that the equivalence result is of independent interest and can contribute to the study of the complexity of distributed graph problems in expander graphs.  Much like the significance of \emph{network decomposition} in the \local model, expander routing stands out as the only nontrivial technique in the design of distributed graph algorithms on expanders in the \congest model. Akin to the theory of {\sf P}-{\sf SLOCAL}-completeness developed in~\cite{ghaffari2017complexity}, an interesting research direction is to explore the possibility of identifying a wide range of fundamental distributed problems on expanders that are equivalent to expander routing.

\section{Preliminaries}
Let $n$ denote the number of vertices and $\Delta$ be the maximum degree. Throughout the paper, we assume our graph has a constant maximum degree,~i.e., $\Delta = O(1)$. In \Cref{sec:general_graphs}, we will show a reduction from general graphs to constant degree graphs. We state some definitions and some basic properties here.

\setlength\multicolsep{0pt}

\paragraph{Conductance} Consider a graph $G=(V,E)$. Given a vertex set subset $S$, define $\vol(S) = \sum_{v \in S} \deg(v)$. Let $\delta(S) = \{(u,v) \mid u \in S, v \in V\setminus S \}$. The {\it conductance} of a cut $S$ and that of a graph $G$ are defined as follows. 
\begin{align*}
\Phi(S) &=\frac{|\delta(S)|}{\min(\vol(S), \vol(V\setminus S))}
&
\Phi(G) &=\min_{\substack{S \subseteq V\\ S\neq \emptyset \mbox{ and } S \neq V}} \Phi(S)
\end{align*}

\paragraph{Sparsity} The {\it sparsity} of a cut $S$ and that of a graph $G$ are defined as follows.
\begin{align*}
\Psi(S) &=\mathmakebox[0pt][l]{\frac{|\delta(S)|}{\min(|S|, |V\setminus S|)}}\phantom{\frac{|\delta(S)|}{\min(\vol(S), \vol(V\setminus S))}}
&
\Psi(G) &=\mathmakebox[0pt][l]{\min_{\substack{S \subseteq V\\ S\neq \emptyset \mbox{ and } S \neq V}} \Psi(S)}\phantom{\min_{\substack{S \subseteq V\\ S\neq \emptyset \mbox{ and } S \neq V}} \Phi(S)}
\end{align*}
We remark that the sparsity $\Psi(G)$ of a graph $G$ is also commonly known as \emph{edge expansion}.

\paragraph{Diameter} Given a graph $G=(V,E)$. For $u,v \in V$, let $\dist_{G}(u,v)$ denote the distance between $u$ and $v$ in $G$. The diameter $D$ is defined to be $D(G) = \max_{u,v \in V(G)}\dist_{G}(u,v)$. The following upper bound on the diameter can 
be obtained by a standard ball-growing argument:
\begin{fact}
Let $G$ be a graph with conductance $\phi$. The diameter $D(G)$ is upper bounded by $O(\phi^{-1} \log n)$.
\end{fact}

\paragraph{Expander Split} The {\it expander split} $G^{\diamond}$ of $G=(V,E)$ is constructed as follows:
\begin{itemize}
\item For each $v \in V$, create an expander graph $X_v$ with $\deg(v)$ vertices with $\Delta(X_v) = \Theta(1)$ and $\Phi(X_v) = \Theta(1)$.

\item For each $v \in V$, fix an arbitrary ranking of the edges incident to $v$. Let $r_v(e)$ denotes the rank of $e$ in $v$. For each edge $e=uv \in E$, add an edge between the $r_{u}(e)$'th vertex of $X_u$ and the $r_v(e)$'th vertex of $X_v$.
\end{itemize}
The expander split will be used to do the reduction from general graphs to constant degree graphs in \Cref{sec:general_graphs}. A key property is that $\Psi(G^{\diamond}) = \Theta(\Phi(G))$. The proof, as well as more properties on expander split, can be found in \cite[Appendix C]{ChangS20}.

\paragraph{Quality of Paths}
Given a set of paths $\oldmathcal{P}$. The quality of $\oldmathcal{P}$, $Q(\oldmathcal{P})$, is defined to be the $\textsf{congestion} + \textsf{dilation}$ of the set of paths. Notice that the smaller this quantity is, the better \emph{quality} we have. Such a notion has been introduced in \cite{HaeuplerWZ21, HRG22}, as there exist randomized algorithms that route along each path simultaneously in $\tilde{O}(Q(\oldmathcal{P}))$ rounds \cite{LMR94, Ghaffari15}. In the deterministic setting, it is straightforward to execute the routing in $\textsf{congestion} \times \textsf{dilation} \leq Q(\oldmathcal{P})^2$ rounds by spending $\textsf{congestion}$ rounds per edge on the paths.

\begin{fact}\label{fact:routing-along-precomputed-paths}
Let $\oldmathcal{P}$ be a set of precomputed routing paths.
Sending one token along every path $P\in\oldmathcal{P}$ simultaneously can be done in deterministic $Q(\oldmathcal{P})^2$ rounds.
\end{fact}

\paragraph{Embeddings} Given graphs $H_1,H_2$ with $V(H_1) \subseteq V(H_2)$, an embedding of $H_1$ into $H_2$ is a function $f: E(H_1) \to \mathcal{P}(H_2)$ that maps the edges of $H_1$ to $\mathcal{P}(H_2)$, the set of all paths in $H_2$.  The quality of the embedding $Q(f)$ is defined to be the quality of the set of paths $\bigcup_{e \in E(H_1)} f(e)$. As the vertex set of $H_1$ is always a subset of $V(H_2)$, we sometimes specify $H_1$ {\it only} by its edge set. 

For the ease of composition, given an embedding $f$, we tweak it so that it can map paths in $H_1$ to paths $H_2$  by defining $f(e_1, \ldots, e_l) = (f(e_1), \ldots , f(e_l))$ for $(e_1, \ldots, e_l) \in \mathcal{P}(H_1)$. Given an embedding $f$ that embeds $H_1$ onto $H_2$ and an embedding $g$ that embeds $H_2$ onto $H_3$, $(g \circ f)$ is an embedding of $H_1$ into $H_3$. 

 Given embedding $f$ that embeds $H_1$ to $G_1$ and embedding $g$ that embeds $H_2$ to $G_2$ with $V(H_1)\cap V(H_2) =  \emptyset$, the embedding $(f \cup g): E(H_1 \cup H_2) \to \mathcal{P}(G_1 \cup G_2)$ is defined to be $$(f\cup g)(e) = \begin{cases}f(e) & e \in E(H_1) \\  g(e) & e \in E(H_2) \end{cases}$$

\paragraph{Matching Embedding}
The following result, developed in \cite{ChangS20,HHS23}, allows us to embed a matching between $S$ and $T$, where $S$ and $T$ are two disjoint subsets:
\begin{lemma}\label{lem:deterministic_path_embedding}
Consider a graph $G=(V,E)$ with maximum degree $\Delta = \polylog(n)$ and a parameter $0 < \psi < 1$. Given a set of source vertices $S$ and a set of sink vertices $T$ with $|S| \leq |T|$, there is a deterministic algorithm that finds a cut $C$ and an embedding $f_{M}$ of a matching $M$ between $S$ and $T$ saturating $S$ with the following requirement in $2^{O(\sqrt{\log n})} \cdot \poly(1/\psi)$ rounds.
\begin{itemize}[itemsep=0pt]
\item {\textbf{Matching:}} The embedding $f_{M}$ has quality $O(\psi^{-2}) \cdot \polylog(n)$.
\item {\textbf{Cut:}} Let $S'\subseteq S$ and $T' \subseteq T$ be the subsets that are not matched by $M$. If $S' \neq \emptyset$, then $C$ satisfies $S' \subseteq C$, $T' \subseteq V \setminus C$, and $\Psi(C) \leq \psi$; otherwise $C = \emptyset$. 
\end{itemize}
\end{lemma}

%% file: hierarchical_decomposition.tex
\section{The Hierarchical Decomposition}
Consider a constant degree graph $G=(V,E)$.
Chang and Saranurak~\cite{ChangS20} gave an algorithm that either finds a balanced sparse cut $C$ with $\Psi(C) \leq \psi$ and $|C| \geq |V|/4$ or finds a subset of vertices $W\subseteq V$ such that $\Psi(G[W]) \geq \log^{-O(1/\epsilon)} n \cdot \poly(\psi)$ with $|W| \geq (2/3)\cdot |V|$, where $0 < \epsilon <1$ is a parameter that the running time depends on. In the latter case, it also produces a hierarchical decomposition $\T$, whose property we summarize in \Cref{prop:hierarchical}. We set $\psi = \Psi(G)/2$ to force it to go into the latter case, as no cut $C$ with $\Psi(C) \leq \psi /2$ can be found. 

\begin{restatable}{property}{hierarchical}
\label{prop:hierarchical} Each node of $\T$ is a vertex set $X \subseteq V$. The root of the tree is a vertex set $W$ with $|W| \geq (2/3)\cdot |V|$. A node of $T$ can be either {\it good} or {\it bad}.  The number of levels $\ell(\T)$ in the hierarchy is upper bounded by $O(1/\epsilon)$. Moreover:

\begin{enumerate}[leftmargin=*]

\item\label{itm:hierarchical:structure}  Let $k = |V(G)|^{\epsilon}$. If a node is good, then it is either {\it terminal} or {\it internal}. A bad node or a terminal good node has no children. A good internal node $X$ consists of a number of good children $X_1 \ldots X_t$, where $(2/3)\cdot k \leq t \leq k$ and they can be ordered so that $\max_{x\in X_i} \ID(x) \leq \min_{y \in X_{i+1}} \ID(y)$ for $1 \leq i < t$. Moreover, it has the same number of bad children $X'_1 \ldots X'_t$. Let $X^{*}_i = X_{i} \cup X'_{i}$. We have $X = X^{*}_1 \cup \ldots \cup X^{*}_t$. There exists $\tau = \Theta(|X|/k)$ such that for each $i$, 
\begin{align*}
\frac{1}{3} \cdot \frac{|X|}{k} \leq |X^{*}_i| \leq 6\cdot \frac{|X|}{k}& &\mbox{and} & & \frac{2}{3} (\tau - 1) \leq |X^{*}_i| \leq 2 \cdot (\tau + 1) 
\end{align*}

\item \label{itm:hierarchical:embedding} Let $p(X)$ denote the parent node of $X$. If a non-root node $X$ is good then it is also associated with a virtual graph $H_X$ with maximum degree $O(\log n)$ whose vertex set is $X$, and an embedding $f_X$ that embeds $H_X$ to $H_{p(X)}$.  The root $X$ is associated with the virtual graph $H_X = G[X]$ with $f_{X}(e) = e$.  

Suppose $X \in \T$ is a good internal node. Let $X_1 \ldots X_t$ be the good children of $X$. The embedding $\bigcup_{i=1}^{t} f_{X_i}$ that embeds $H_{X_1} \cup \ldots \cup H_{X_t}$ onto $H_{X}$ has quality $\polylog(n) \cdot O(\psi^{-1})$ in $X$ if $X$ is the root, and $\polylog(n)$ otherwise.

In addition, for any good node $X$, $\Psi(H_{X}) = \poly(\psi) \cdot \log^{-O(1/\epsilon)} n$ if $X$ is the root and $\Psi(H_{X}) = \Omega(1/\log^{\Theta(\ell(\T) - \ell(X))} n) = \log^{-O(1/\epsilon)} n$ otherwise, where $\ell(X)$ is the level that $X$ is at in the hierarchy (the root has level 0).

\item \label{itm:hierarchical:matching} Suppose that $X$ is a good internal node. For each $H_{X_i}$, it can be extended to a virtual graph $H^{*}_{i}$ of $X^{*}_i$ by adding a matching $M^{*}_i$ between $X_i$ and $X'_i$ to $H_{X_i}$ such that each vertex in $X'_i$ is matched. This also implies $|X'_i| \leq |X_i|$ and so 
$$|X_1 \cup \ldots \cup X_t| \geq |X|/2$$  Moreover, there exists an embedding $f_{M_{X}}$ that embeds $\bigcup_{i=1}^{t} M^{*}_i$ onto $H_X$ with quality $\polylog(n)$ if $\ell(X) \geq 1$, and quality of $O(\psi^{-1})\cdot \polylog(n)$ if $\ell(X) = 0$.

\end{enumerate}
\end{restatable}
\Cref{prop:hierarchical}(\ref{itm:hierarchical:structure}) says that every part $X^{*}_i$ has roughly the same size, with up to a constant factor difference. \Cref{prop:hierarchical}(\ref{itm:hierarchical:embedding}) describes the embedding inside each $X_i$.  \Cref{prop:hierarchical}(\ref{itm:hierarchical:matching}) describes the embedding between $X_i$ and $X'_i$. See \Cref{fig:hierarchical_decomposition} for an illustration of the decomposition. 

\begin{figure}[t]\centering
\includegraphics[width=\textwidth]{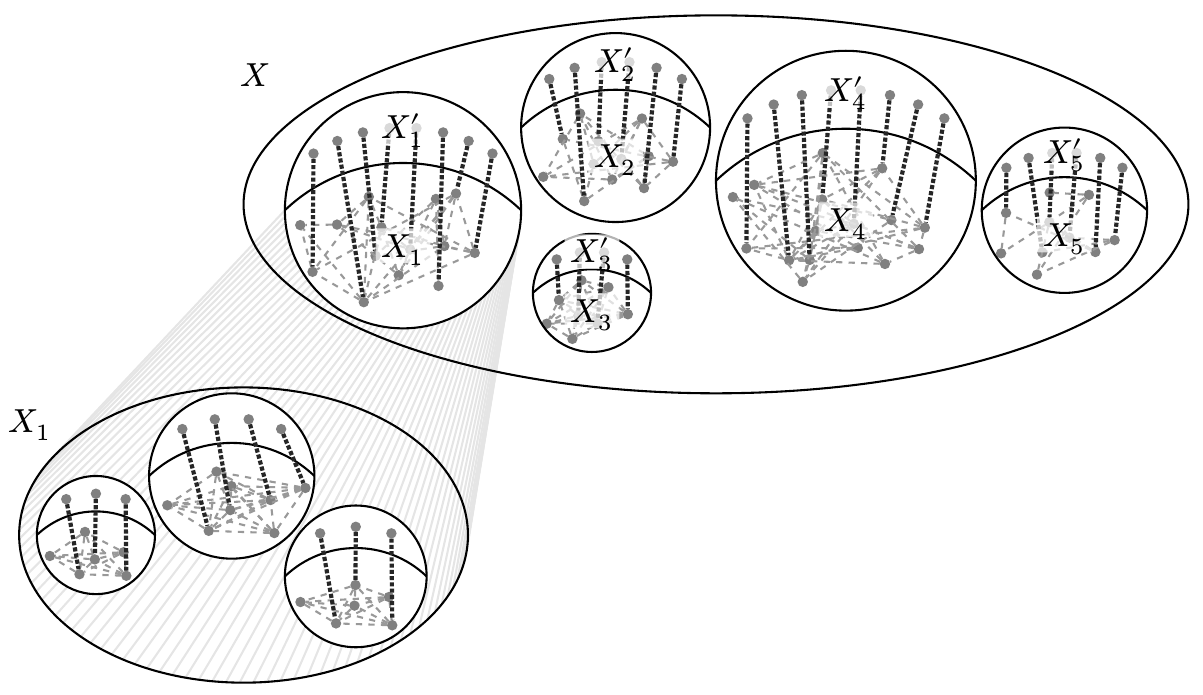}
\caption{An illustration of the hierarchical decomposition. The gray dotted edges denote the expander embedding as described in  \Cref{prop:hierarchical}(\ref{itm:hierarchical:embedding}). For example, the gray dotted edges inside $X_1$ is the virtual graph $H_{X_{1}}$. The base graph of the child node with vertex set $X_1$ is now $H_{X_1}$. The black dotted edges between $X_i$ and $X'_i$ form a matching embedding described in \Cref{prop:hierarchical}(\ref{itm:hierarchical:matching}).}\label{fig:hierarchical_decomposition}
\end{figure}

Some properties listed above may not be explicitly stated in~\cite{ChangS20}.
Thus, for the sake of completeness we will go over the construction of \cite{ChangS20} to verify these properties in \Cref{sec:hierarchical_construction}.
\begin{theorem}[{\cite{ChangS20}}]\label{thm:build-hierarchy}
Let $G$ be a constant degree $\phi$-expander and $k= n^{\epsilon}$ be a parameter. Then, there exists a deterministic \congest algorithm that computes a hierarchical decomposition $\T$ that satisfies \Cref{prop:hierarchical} in $\poly(\phi^{-1}) \cdot (n^{O(\epsilon)} +  \log^{O(1/\epsilon)} n)$ time.
\end{theorem}

\begin{definition}
Let $X \in \T$ be a good node whose level is $\ell(X)$. The {\it flatten embedding} $f^{0}_{X}$ is an embedding that embeds $H_X$ to $G$, defined as 
$$f^{0}_{X} = f_{p^{(\ell(X))}(X)} \circ \ldots \circ f_{p^{(2)}(X)} \circ f_{p(X)}\circ f_{X}$$
\end{definition}

\begin{corollary}
For each $X \in \T$, let $\oldmathcal{P}_{X}$ be any collection of paths in $X$. Suppose that the quality of each $\oldmathcal{P}_{X}$ is upper bounded by $Q$. Let $\oldmathcal{P}'= \bigcup_{X \in \T} f^{0}_{X}(\oldmathcal{P}_{X})$ be the flatten mapping of these paths to $G$. We have that $Q(\oldmathcal{P}') =  Q \cdot \poly(\psi^{-1}) \cdot \log^{O(1/\epsilon)} n$.
\end{corollary}

\begin{proof}
Let $\T_i = \{X \in \T \mid \ell(X) = i \} $. Define $f^{i} = \bigcup_{X\in \T_i} f_{X}$ to be the union of embedding from level-$i$ nodes to level-$(i-1)$ nodes. By \Cref{prop:hierarchical}(\ref{itm:hierarchical:embedding}), $Q(f^{i}) = \polylog n$ if $i > 1$ and $Q(f^{i}) = \poly(\psi^{-1}) \cdot \polylog n$ otherwise. Since $\bigcup_{X \in \T_i} f^{0}_{X}(P_{X}) = (f^{1} \circ \ldots \circ f^{i})(\bigcup_{X\in \T_i} \oldmathcal{P}_X)$, we have $Q(\bigcup_{X \in \T_i} f^{0}_{X}(\oldmathcal{P}_{X})) = Q \cdot O(\psi^{-1}) \cdot \log^{O(i)} n$. Summing this over each $i=1,\ldots, O(1/\epsilon)$, we conclude that the quality of $\oldmathcal{P}'$ is at most $Q \cdot O(\psi^{-1}) \cdot \log^{O(1/\epsilon)} n$.
\end{proof}

\paragraph{Embedding a Matching To Cover the Whole Graph}
We note that the root $W \in \T$ does not cover all the vertices in $V$. Using \Cref{lem:deterministic_path_embedding}, we can pre-embed a matching between vertices of $V\setminus W$ and $W$ with good quality so that tokens can be routed to the hierarchy easily.

\begin{lemma}\label{lemma:root-level-matching}
Let $W \in \T$ be the root of the hierarchical decomposition. There exists a \congest algorithm that finds an embedding $f_{M_{root}}$ of a matching $M_{root}$ between $V\setminus W$ and $W$ that saturates $V \setminus W$ with $Q(f_{M_{root}})=\psi^{-2}\cdot \log^{O(1)} n$ and the runtime is $2^{O(\sqrt{\log n})} \cdot \poly(\psi^{-1})$.
\end{lemma}
\begin{proof} Note that $|W| \geq (2/3)|V|$ by \Cref{prop:hierarchical}. We set $S=W$ and $T=V \setminus W$ and so $|S| < |T|$. We then apply \Cref{lem:deterministic_path_embedding} with $\psi= \Psi(G)/2$ so that it returns a matching embedding with the desired quality. \end{proof}

\paragraph{Leaf Trimming}
We will trim the leaves of $\T$ so that every leaf node contains at least $k^4 = O(n^{4\epsilon})$ vertices. This can be done level by level from the last level. Each level takes $O(D' + k^4)\cdot \log^{O(1/\epsilon)} n $ rounds, where $D'=\log^{O(1/\epsilon)} n$ is a diameter upper bound of $H_{X}$ of each $X \in \T$ in that level.

\begin{definition} Given $X \in \T$, define $X_{best} \subseteq X$ to be the union of the good leaf nodes in the subtree rooted at $X$. 
\end{definition}

\begin{definition}
Define $\rho_{best} = \max_{X \in \T} |X|/|X_{best}|$. 
\end{definition}

Note that with \Cref{prop:hierarchical}, we have $\rho_{best}=2^{O(1/\epsilon)}$.

\section{Reducing to Internal Routing Tasks}

We first consider the core setting of the expander routing problem, where the input graph $G$ is a constant degree expander of sparsity $\psi$.
The main task described below summarizes the routing task on $G$:

\begin{definition}[\textbf{Task 1}]
Let $G$ be a constant degree $\psi$-expander, where each vertex of $G$ has a unique destination ID in $\{1, 2, \ldots, n^{O(1)}\}$.
Let $L$ be a parameter that depicts the maximum load.
Suppose that each node in $G$ holds at most $L$ tokens, and each node is the destination of at most $L$ tokens. The goal is to route the tokens to their destinations.
\end{definition}

However, as the leaves (i.e.,~the best nodes) of our hierarchical decomposition do not cover the whole graph, it would be difficult to solve \textbf{Task 1} directly. Instead, we consider a routing problem where all the destinations of the tokens are on the best nodes, specified by their ranks. In \Cref{sec:reductions}, we show how to reduce {\bf Task 1} to the following task by delegating it to the best nodes and having them do expander sorting:

\begin{definition}[\textbf{Task 2}]\label{task2} Let $X$ be a good node in the hierarchical decomposition $\T$ of the input constant degree $\psi$-expander $G$.
Let $L$ be a parameter.
Suppose that 
each node holds for at most $L$ tokens.
Each token $z$ has a \emph{destination marker} $i_z$
and there are at most $L\rho_{best}$ tokens for each destination marker $i_z$.
The goal of the task is to route all tokens with destination marker $i_z$ to the $i_z$-th smallest vertex among $X_{best}$.
\end{definition}

Note that as {\bf Task 2} will be solved recursively, we defined the task on every component $X$ of the hierarchy $\T$. We will now focus on solving {\bf Task 2} by using the ideas discussed in the introduction. In the following, We identify the key task for solving {\bf Task 2} recursively. %

Let $X\in\T$ be an internal component and let $X_1^*, X_2^*, \ldots, X_t^*$ be the parts of $X$ derived from \Cref{thm:build-hierarchy}.
We note that with broadcasts, it is possible for every vertex $v\in X$ obtaining the number of best vertices within all its parts during preprocessing in $O((k+D(H_{X}))\cdot Q(\bigcup_{X' \in \oldmathcal{T}}f^{0}( H_{X'}))) = \poly(\psi^{-1}, k, \log^{1/\epsilon}n)$ rounds for all $X\in\oldmathcal{T}$ in parallel. Furthermore, 
by \Cref{prop:hierarchical}(\ref{itm:hierarchical:structure}), the IDs of the vertices in $X_{best}$ are partitioned in the sorted order.
This allows the algorithm to rewrite the destination marks at the beginning of handling a query on $X$:
For any token $z$ with destination mark $i_z$,
the algorithm computes two values $(j_z, i'_z)$, where $j_z\in \{1, 2, \ldots, t\}$ is the index of the part containing the $i_z$-th smallest best vertex, and $i'_z = i_z - \sum_{j < j_z} |X_{best}\cap X_j^*|$ is the next-level destination mark.

Therefore, to solve \textbf{Task 2} on $X$, it suffices to first route all tokens $z$ to any vertex in the part $X_{j_z}^*$.
Finally, for any part $X_j^*=X_j\cup X'_j$, through \Cref{prop:hierarchical}(\ref{itm:hierarchical:matching}) we are able to route all tokens $X_{j}^*$ to $X_j$ such that a next-level \textbf{Task 2} can be called.
We summarize this task as follows.

\begin{definition}[\textbf{Task 3}]
Let $X\in\T$ be a non-leaf good node and $X_1^*, X_2^*, \ldots, X_t^*$ be its parts.
Each vertex holds at most $L$ tokens for some parameter $L$. 
For each token $z$ there is a \emph{part mark} $j_z$.
Suppose that for each $j\in \{1, 2, \ldots, t\}$ there are at most $L\cdot |X_j^*|$ tokens having the same part mark $j$.
The task is accomplished whenever every token $z$ with a part mark $j_z$ is located at a vertex in $X_{j_z}^*$ and each vertex holds at most $2L$ tokens.
\end{definition}

In the following sections, we introduce tools and aim to give algorithms for {\bf Task 3}.

%% file: recursive_routing.tex
\section{Core Tools: Shuffler and Expander Sorting}
\label{sec:cut-matching}

As mentioned in \Cref{sec:our-approach}, \textbf{Task 3} is solved by routing all tokens into a dispersed configuration.
The tokens are routed through a shuffler.
In \Cref{sec:routing-infrastructure}, we describe an algorithm that constructs such a shuffler.
Our algorithm implements R\"acke, Shah, and T\"aubig's cut-matching game \cite{RST14} but with a twist in order to be constructed efficiently.
The efficiency comes from the fact that we already have the precomputed hierarchical decomposition $\T$.
To process queries using the constructed shuffler, it is necessary to route the tokens to specified \emph{shuffler portals}.
This can be implemented by expander sorting procedures.
In \Cref{sec:expander-sort} we introduce the expander sorting and convenient procedures that can be applied to solving \textbf{Task 3}.

\subsection{Shuffler}\label{sec:routing-infrastructure}

In this subsection, we define shufflers and the algorithm for constructing them during the preprocessing time. Consider a good node $X \in \T$ with $|X| \geq n^{4\epsilon}$. Let $X^{*}_1 \ldots X^{*}_t$ be a partition of $X$, where $X^{*}_i = X_i \cup X'_i$ as defined in \Cref{prop:hierarchical}.

\paragraph{Cut-Matching Game and Shuffler}
To achieve this, we run a variant of cut-matching game~\cite{KRV09}. The cut-matching game consists of a cut player and a matching player, and they alternatively pick a cut and add a matching to an initially empty graph.  In our variant, the cut-matching game is ``played'' on the \emph{cluster graph} $Y$, which is the multigraph by contracting each of the $t$ parts of $X$. In each iteration $q$ of the cut-matching game, the cut player first obtains a cut on $Y$, which implies a cut on $X$.

Then, the matching player finds an embedding $f_{M_{X}^q}$ of a virtual matching $M_X^q$ on $X$, which we will show how to transform to a natural \emph{fractional matching} on $Y$. The sequence of all computed matchings and embeddings $\mathcal{M}_X := ((M_X^1, f_{M_{X}^{1}}), (M_X^2, f_{M_{X}^{2}}), \ldots, ({M}_X^\lambda, f_{M_{X}^{\lambda}}))$ is then called the \emph{shuffler}. Now, we formally define the aforementioned terms. %

\begin{definition}\label{def:cluster-graph}
Let $Y$ denote the cluster graph obtained by contracting each $X^{*}_i$ in $H_{X}$. The set of vertices $V(Y)=\{v_1, v_2, \ldots, v_t\}$ has exactly $|Y|=t$ vertices where $v_i$ corresponds to the vertex contracted from $X^*_i$.
Given a subset $S \subseteq V(Y)$, let $S_{X}$ denote the corresponding vertex set $\cup_{i: v_i\in S} X_i^*$ in $X$.
\end{definition}

A \emph{fractional matching} $M=\{x_{uv}\}$ of a graph $Y$ is a mapping that maps each unordered pair $\{u, v\}\in \binom{V(Y)}{2}$ to a real number $x_{uv}\in[0, 1]$ such that for all $u\in V(Y)$, the fractional degree is at most one: $\sum_v x_{uv} \le 1$. Given a matching $M_{X}$ in $X$, the corresponding {\it natural fractional matching} $M = \{x_{uv}\}_{(u,v) \in \binom{V(Y)}{2}}$ is $Y$ of defined to be:
$$x_{uv} = \frac{|\{(a,b) \in M_{X} \mid a \in X^{*}_{u}, b \in X^{*}_{v} \}|}{n'} \mbox{, where $n' = 6|X|/k$}.$$

\begin{definition}
Given any fractional matching $M = \{x_{uv}\}$ on $Y$, we define a $t\times t$ matrix $R_M$ with
\[R_{M}[i,j] := \begin{cases} \frac12 + \frac12 \cdot (1 - \sum_{k \neq i} x_{v_i v_k}) & \mbox{if $i = j$,} \\
\frac12 \cdot  x_{v_i v_j} & \mbox{if $i \neq j $.} 
\end{cases}\]
Let $(M^1, \ldots, M^i)$ be a sequence of fractional matchings.
If the context is clear, we omit the sequence and denote by $R_i$  the product of matrices $R_{M^{i}} \cdots R_{M^{2}}R_{M^1}$.
For any vertex $y\in V(Y)$, let $R_i[y]$ be the row vector in $R_i$ that corresponds to the vertex $y$. For $a, b\in V(Y)$, the $b$-th entry of $R_i[a]$ can be interpreted as the probability of a random walk that starts from $b$ and ends up at $a$. It is straightforward to verify that all entries of $R_i[y]$ adds up to $1$ for all $y\in V(Y)$.
\end{definition}

Let $\bm{1}$ be the all-one vector, and for brevity, we denote $\bm{\frac{1}{|Y|}} := \frac{1}{|Y|}\bm{1}$. The following definition sets up a potential function for the cut-matching game.
Let $\|\cdot\|$ to be the standard 2-norm function for a vector.

\begin{definition}[\cite{KRV09}]
Let $Y$ be a cluster graph defined in \Cref{def:cluster-graph} and let $(M^1, \ldots, M^i, \ldots)$ be a sequence of fractional matchings on $Y$.
We define the potential function $\Pi(i) := \sum_{y \in Y(V)} \|R_{i}[y]  - \bm{\frac{1}{|Y|}} \|^2$. 
\end{definition}

\begin{definition}[Shuffler]\label{def:shuffler} 
Given $X\in\T$, a \emph{shuffler} of $X$ consists of a sequence of $\lambda$  matching embeddings $\mathcal{M}_X:=((M_X^1, f_{M_{X}^{1}}), (M_X^2, f_{M_{X}^{2}}), \ldots, ({M}_X^\lambda, f_{M_{X}^{\lambda}}))$ on $X$ such that if $(M_1, \ldots, M^{\lambda})$ is the corresponding fractional matching of $(M_X^1, \ldots, M_{X}^{\lambda})$ in $Y$, the random walk induced by it nearly mixes, as characterized by the following bound on the potential function:
$$\sum_{y \in V(Y)} \left\lVert R_{i}[y]  - \bm{\frac{1}{|Y|}} \right \rVert^2 \leq \frac{1}{9n^3}$$

In addition, for each $i$, $1\le i\le \lambda$, the embedding $f_{M_{X}^{i}}$ has quality $\log^{O(1/\epsilon)} n$ in $X$ for non-root $X$, and $\poly(\psi(G)^{-1})\log^{O(1/\epsilon)} n$ if $X$ is the root. The quality of the shuffler, which essentially has the same order of magnitude as the quality of each embedding $f_{M_{X}^{i}}$, is defined to be $$Q(\mathcal{M}_{X}) := Q\left(\bigcup_{i=1}^{\lambda}f_{M_{X}^{i}}(M_{X}^{i})\right).$$
\end{definition}

We prove the following lemma in \Cref{sec:the-modified-cut-matching-game}.
\begin{lemma}\label{lem:disperse}
There exists an \congest algorithm such that, given a good node $X\in\T$, computes a shuffler of $X$ in $\poly(\psi^{-1}, k, \log^{1/\epsilon} n)$ rounds.
Moreover, the shuffler has $\lambda=O(\log n)$ matching embeddings with quality $Q(\mathcal{M}_X)=\poly(\psi^{-1}, \log^{1/\epsilon} n)$.
\end{lemma}

\subsection{Distributed Expander Sorting}\label{sec:expander-sort}

In this subsection, we introduce several primitives that are recursively dependent on the internal routing tasks (i.e., \textbf{Task 2} and \textbf{Task 3}). Perhaps, the most interesting side-product result we obtain is a deterministic sorting algorithm on an expander graph, described as follows.

\begin{restatable}[Deterministic Expander Sort]{theorem}{deterministicexpandersort}
\label{thm:expander-sort}
Let $X^* = X\cup X'$ be a virtual graph such that $X$ is a good node in the hierarchical decomposition $\T$ of a $\psi$-sparsity expander and there is an embedded $X'$-matching $f_M$ from $X'$ to $X$ with a flattened quality $Q(f^0_M)=\poly(\psi^{-1}, \log^{1/\epsilon}n)$.
Suppose that each node holds at most $L$ tokens, and each token $z$ has a (not necessarily unique) key $k_z$. Then, there exists a \congest algorithm such that, when the algorithm stops, for any two tokens $x$ and $y$ on two different vertices $u$ and $v$ with $\mathit{ID}(u) < \mathit{ID}(v)$, we have $k_x \le k_y$. Moreover, each vertex holds at most $L$ tokens.
The preprocessing time satisfies the following recurrence relations:
\begin{align*}
T_{\rm sort}^{\textsf{pre}}(|X^*|) &= 2Q(f^0_{M_{root}})^2 + T_{\rm sort}^{\textsf{pre}}(|X|)\\
T_{\rm{sort}}^{\textsf{pre}}(|X|) &= \begin{cases}
T_2^{\textsf{pre}}(|X|) + O(\log n)\cdot T_2(|X|, 1) + \poly(\psi^{-1}, k, \log^{1/\epsilon}n) & \text{if $X$ is non-leaf,}\\
\poly(\psi^{-1}, k, \log^{1/\epsilon}n)  & \text{if $X$ is a leaf.}
\end{cases}
\end{align*}
The query time satisfies the following recurrence relations:
\begin{align*}
    T_{\rm{sort}}(|X^*|, L) &= 2Q(f^0_{M_{root}})^2 + T_{\rm{sort}}^{\textsf{pre}}(|X|)\\
    T_{\rm sort}(|X|, L) &= \begin{cases} 
T_3(|X|, L) + L \rho_{best}\cdot  Q(\oldmathcal{I}_{\textsf{AKS}})^2 + 
L \cdot Q(f^0_{M_X})^2 + 
T_{\rm{sort}}(6|X|/k, L)  \hspace*{0.1cm} \text{if $X$ is non-leaf,}\\
L \cdot \poly(\psi^{-1}, \log^{1/\epsilon} n) 
\hspace*{4.365cm}\text{if $X$ is a leaf component.}
\end{cases}
\end{align*}
\end{restatable}

Solving the recurrence relation requires solving the recurrence relations for \textbf{Task 2} and \textbf{Task 3}, which we defer to \Cref{sec:solving-recurrence-relation}.

\paragraph{Applications}
The distributed expander sort can be used for the following useful primitives, including token ranking, local propagation, local serialization, and local aggregation.
The term \emph{local} here refers to the flexibility of setting an arbitrary \emph{grouping key}, such that the described tasks are performed on each group of tokens independently but simultaneously.
We state the results here and provide detailed proofs in \Cref{sec:distributed-expander-sorting-appendix}.

\begin{restatable}[Token Ranking]{theorem}{tokenranking}\label{thm:token-ranking}
In $O(T_{\rm sort}(|X^*|, L))$ rounds,
each token receives a rank $r_z$ which equals the number of distinct keys that are strictly less than $k_z$.
\end{restatable}

\begin{restatable}[Local Propagation]{lemma}{localpropagation}\label{lemma:local-propagation}
Suppose that each token has a key $k_z$, a unique tag $u_z$, and a variable $v_z$.
In $O(T_{\rm sort}(|X^*|, L))$ rounds, each token's variable is rewritten as $v_{z^*}$ where $z^*=\mathrm{argmin}_x\{u_x \ |\ k_x=k_z\}$.
\end{restatable}

\begin{restatable}[Local Serialization]{corollary}{localserialization}\label{lemma:local-serialization}
In $O(T_{\rm sort}(|X^*|, L))$ rounds, each token receives a distinct value $\mathit{SID}_z \in \{0, 1, \ldots, \mathit{Count}(k_z)-1\}$ among all tokens with the same key. Here $\mathit{Count}(k_z)$ refers to the number of tokens with key $k_z$.
\end{restatable}

\begin{restatable}[Local Aggregation]{corollary}{localaggregation}\label{lemma:local-aggregation}
In $O(T_{\rm sort}(|X^*|, L))$ rounds, each token $z$ learns the value $\mathit{Count}(k_z)$.
\end{restatable}

\section{Solving Task 2 and Task 3 using Shufflers}
\label{sec:solving-task2-querytime}

In this section, we aim to describe our algorithms for solving \textbf{Task 2} and \textbf{Task 3}.

\paragraph{Algorithm for Task 2} 
The algorithm follows immediately by \Cref{lemma:leaf-components} and \Cref{thm:task3}: the algorithm routes the tokens to the corresponding parts, sends the tokens to vertices along the matching toward the good node, and then recurse on the good nodes.

\paragraph{Algorithm for Task 3}
To solve Task 3, we use another meet-in-the-middle idea.
Suppose that we are able to disperse all tokens with the same part mark
as even as possible.
Then, the algorithm may apply the same procedure on the instance where $2L$ dummy tokens are created with part mark $j$ on each vertex at $X_j^*$. Once these dummy tokens are dispersed and meet the real token of the same part mark,
a desired routing is found and each dummy token brings \emph{at most one} real tokens to the goal.

\begin{definition}\label{def:dispersed-configuration}
Let $X\in\T$ be a non-leaf good node with $t$ parts $X_1^*, X_2^*, \ldots, X_t^*$.
We say that a configuration is a \emph{dispersed configuration}, if for all $i, j\in\{1, 2, \ldots, t\}$, the number of tokens on vertices of $X_i^*$ having part mark $j_z=j$, denoted as $|T_{i, j}|$, satisfies 
\[0.9\frac{N_j}{t} - 0.1\frac{|X|}{t^2} \le |T_{i, j}| \le 1.1\frac{N_j}{t} + 0.1\frac{|X|}{t^2}, \]
where $N_j$ is the total number of tokens whose part mark equals $j$.
We say that a configuration is a \emph{final configuration}, if every token with part mark $j_z=j$ is located on a vertex in $X_j^*$.
\end{definition}

Notice that the condition of \textbf{Task 3} requires that the number of real tokens with any part mark $j$ is at most $L\cdot |X_j^*|$.
The total number of dummy tokens generated from part $j$ is $2L\cdot |X_j^*|$.
With the above definition \Cref{def:dispersed-configuration},
if both real tokens and dummy tokens are routed to a dispersed configuration, we will show that on each part $X_i^*$ and for each part mark $j$ the total number of dummy tokens is guaranteed to be outnumbered than the number of real tokens of the same part mark.

Suppose that we have already applied the above idea where the real tokens and the dummy tokens are routed into dispersed configurations.
There is a caveat: these tokens may be located at different vertices within the same part $X_i^*$ and they do not meet each other.
Therefore, to complete the route, we will have to match the real tokens and the dummy tokens of the same part mark within each $X_i^*$. 

In \Cref{sec:arbitrary-conf-to-dispersed-conf} and \Cref{sec:tokens_to_portal} we show how to transform an arbitrary configuration to a dispersed configuration recursively. In \Cref{sec:merge_dispersed}, we show how to merge two potentially different dispersed configurations. In \Cref{sec:task2-leaf-components}, we show how to solve a leaf case of \textbf{Task 2}. We finish the section with time complexity analysis in \Cref{sec:round_analysis}.

\subsection{Arbitrary Configuration to the Dispersed Configuration}\label{sec:arbitrary-conf-to-dispersed-conf}
Let $\mathcal{M}_X := ((M_X^1, f_{M_{X}^{1}}), (M_X^2, f_{M_{X}^{2}}), \ldots, ({M}_X^\lambda, f_{M_{X}^{\lambda}}))$ be the shuffler of $X$, where $\lambda = O(\log n)$. Let $(M^{1}, \ldots, M^{\lambda})$ be the sequence of corresponding natural fraction matching in $Y$ to the sequence of matching $(M_X^1, M_X^2, \ldots, M_{X}^{\lambda})$.
Recall from \Cref{def:shuffler} that the random walk $R_{M^{\lambda}} \ldots R_{M^{1}}$ converges to nearly uniform distribution from any initial distribution.
We will distribute the tokens according to the fraction matching iteration by iteration. That is, in iteration $q$, we will distribute the tokens according to $F_{M^{q}}$.  

Consider a fractional matching $M$ in an iteration $s$. Let $T_{i,l}$ be the $T_{l}$ tokens at $X^{*}_i$. In each iteration, the goal is the following: For each $i,j,l$, we send $\lfloor (m_{ij}/2)|T_{i,l}| \rfloor$  tokens in $T_{i,l}$ from $X^{*}_i$ to $X^{*}_j$. To achieve this, we will need to first route the $\lfloor (m_{ij}/2)|T_{i,l}| \rfloor$ tokens to the {\it portals} $P_{i,j}$, where $P_{i,j} \subseteq X^{*}_i$ is defined as $P_{i,j} = \{x \in X^{*}_i \mid \mbox{$(x,y) \in M^{q}_X$ for some $y \in X^{*}_j$} \}$.

In \Cref{sec:tokens_to_portal}, we describe how such a task of routing the tokens to the portals can be done recursively. Once the tokens are routed to the portals, they can follow the path embedding corresponding to the virtual matching edge to arrive in $X^{*}_j$. The following lemma 
shows that a dispersed configuration is achieved after doing the token distribution according to $(M^{1}, \ldots, M^{\lambda})$.

\begin{restatable}{lemma}{tokendistribution}
Let $(M^{1}, \ldots, M^{\lambda})$ be the sequence of natural fractional matching in $Y$ that corresponds to the sequence of matching in $\mathcal{M}_{X}$. For $q=1\ldots \lambda$, suppose that during iteration $q$, we send $\lfloor (m_{ij}/2)|T^{q-1}_{i,l}| \rfloor$ tokens in $T^{q-1}_{i,l}$ from $X^{*}_i$ to $X^{*}_j$ for each $1 \leq i,j,l \leq t$, where $T^{q-1}_{i,l}$ is set of tokens in $X^{*}_i$ destined to part $l$ at the end of iteration $q-1$. A dispersed configuration is achieved at the end of the procedure. %
\end{restatable}

\begin{proof}
Fix an iteration $q$ where $1\le q \le \lambda$. We claim that:
$$ R_{M^q} \cdot (|T^{q-1}_{1,l}|, \ldots, |T^{q-1}_{t,l}|)^\intercal - (t,\ldots,t)^\intercal \leq (|T^q_{1,l}|, \ldots, |T^q_{t,l}| )^\intercal \leq R_{M^q} \cdot (|T^{q-1}_{1,l}|, \ldots, |T^{q-1}_{t,l}|)^\intercal + (t,\ldots,t)^\intercal$$

This is because:
\begin{align*} |T^{q}_{i,l}| &= \left(|T^{q-1}_{i,l}|  - \sum_{\substack{j=1\\j\neq i }}^{t} \lfloor(m_{j,i}/2) |T^{q-1}_{i,l}| \rfloor \right) + \sum_{\substack{j=1\\j\neq i }}^{t} \lfloor (m_{j,i}/2) \cdot |T^{q-1}_{j,l}|  \rfloor \\
&\leq t +  \left(|T^{q-1}_{i,l}|  - \sum_{\substack{j=1\\j\neq i }}^{t} (m_{j,i}/2) |T^{q-1}_{i,l}|  \right) + \sum_{\substack{j=1\\j\neq i }}^{t}  (m_{j,i}/2) \cdot |T^{q-1}_{j,l}| \\
&= t + R_{M^q}[i]\cdot (|T^{q-1}_{1,l}|, \ldots, |T^{q-1}_{t,l}|)^\intercal  \end{align*}

Similarly: 
\begin{align*} |T^{q}_{i,l}| &\geq (-t) +  \left(|T^{q-1}_{i,l}|  - \sum_{\substack{j=1\\j\neq i }}^{t} (m_{j,i}/2) |T^{q-1}_{i,l}|  \right) + \sum_{\substack{j=1\\j\neq i }}^{t}  (m_{j,i}/2) \cdot |T^{q-1}_{j,l}| \\
&= (-t) + R_{M^q}[i]\cdot (|T^{q-1}_{1,l}|, \ldots, |T^{q-1}_{t,l}|)^\intercal  \end{align*}

Let $\lambda = O(\log n)$ be the last iteration. We have:
\begin{align*}
(|T^{\lambda}_{1,l}|,\ldots, |T^{\lambda}_{t,l}|)^{\intercal} &\leq R_{M^{\lambda}} \cdot (|T^{\lambda-1}_{1,l}|,\ldots, |T^{\lambda-1}_{t,l}|)^{\intercal}  + (t,\ldots, t)^{\intercal} \\
&\leq  R_{M^{\lambda}} \cdot ( R_{M^{\lambda-1}}\cdot(|T^{\lambda-2}_{1,l}|,\ldots, |T^{\lambda-2}_{t,l}|)^{\intercal}  + (t,\ldots, t)^{\intercal} ) +  (t,\ldots, t)^{\intercal} \\
&= R_{M^{\lambda}} \cdot  R_{M^{\lambda-1}}\cdot (|T^{\lambda-2}_{1,l}|,\ldots, |T^{\lambda-2}_{t,l}|)^{\intercal} + 2(t,\ldots,t)^{\intercal} \\
&...\\
&=R_{M^{\lambda}} \cdot R_{M^{\lambda-1}} \cdot \ldots \cdot  R_{M^{1}}(|T^{0}_{1,l}|,\ldots, |T^{0}_{t,l}|)^{\intercal} + \lambda \cdot (t,\ldots, t)^{\intercal} \\
&\leq \left(\left(\frac{1}{t}+ \frac{1}{n^{1.5}}\right)\cdot |T_{l}|  + \lambda t, \ldots, \left(\frac{1}{t} + \frac{1}{n^{1.5}}\right)\cdot |T_{l}| + \lambda t \right)^{\intercal} \\
&= \left(\frac{N_{l}}{t} + \frac{N_{l}}{n^{1.5}} + \lambda t, \ldots, \frac{N_{l}}{t} + \frac{N_{l}}{n^{1.5}} + \lambda t \right)^{\intercal} \\
&= \left(1.1\frac{N_{l}}{t}  + \lambda t, \ldots, \frac{1.1N_{l}}{t} +  \lambda t \right)^{\intercal}\\
&= \left(1.1\frac{N_{l}}{t}  + 0.1\frac{|X|}{t^2}, \ldots, 1.1\frac{N_{l}}{t} +  0.1\frac{|X|}{t^2} \right)^{\intercal} \\
& \hspace{50mm} 0.1\frac{|X|}{t^2} \geq 0.1\frac{n^{4\epsilon}}{n^{2\epsilon}} \geq 0.1n^{2\epsilon} \geq n^{\epsilon} \cdot O(\log n) \geq t\lambda  
\end{align*}
Similarly, we have:
\begin{align*}
(|T^{\lambda}_{1,l}|,\ldots, |T^{\lambda}_{t,l}|)^{\intercal} &\geq \left(0.9\frac{N_{l}}{t}  - 0.1\frac{|X|}{t^2}, \ldots, 0.9\frac{N_{l}}{t} -  0.1\frac{|X|}{t^2} \right)^{\intercal}
\end{align*}
Therefore, a dispersed configuration is achieved after iteration $\lambda$.
\end{proof}

The following corollary is useful to control the number of tokens in each part, which can be useful for deriving the final recurrence. Due to its similar flavor, we state it here:

\begin{corollary}\label{lemma:number-of-max-tokens-after-matching}
Let $N_{\rm{max}}$ be the maximum number of tokens within any part at the beginning of the routing.
For any $q$ such that $1\le q\le \lambda$, the total number of tokens within the part $X_i^*$ after sending the tokens along fraction matchings $M^1, M^2, \ldots, M^q$ is at most $N_{\rm{max}} + t^2q$.
\end{corollary}

\begin{proof}
This can be done by an induction on $q$.
Let $N_i^q$ be the number of tokens within the part $X_i^*$ after iteration $q$ and let $N_{\rm{max}}^q = \max_i\{N_i^q\}$.
Then, for all $i$ we have $N_i^0\le N_{\rm{max}}$.

We notice that each $M^q$ is a fraction matching.
If the tokens are fractional as well then the total amount of tokens does not change.
However, due to the fact that tokens are integral, we have to carefully upper bound the total number of tokens:
\begin{align*}
N_i^q &\le \bigg( N_i^{q-1}
- \underbrace{\sum_{j=1}^t\sum_{l=1}^t \left\lfloor \frac12 m_{i, j}|T_{i, l}^{q-1}|\right\rfloor}_{\text{tokens sent away}}
\bigg)
 + \underbrace{\sum_{j=1}^t \sum_{l=1}^t \left\lfloor \frac12 m_{j, i}  |T_{j, l}^{q-1}| \right\rfloor}_{\text{tokens received}}\\
 &\le \left(\frac12 N_i^{q-1} + t^2\right) + \frac12 N_{\rm{max}}^{q-1}\\
 &\le N_{\rm{max}}^{q-1} + t^2\ .\qedhere
\end{align*} 
\end{proof}

\subsection{Routing the Tokens to the Portals}\label{sec:tokens_to_portal}

Consider a particular part $X_i^*$ and a particular fraction matching $M^q$.
In this subsection, we describe a subroutine that routes all tokens on $X_i^*$ to designated portals $P_{i, j}$.
The goal of this subroutine is that for all $l$ and for all tokens with the same specific part mark $l$, there will be exactly $\lfloor (m_{i,j}/2)|T_{i, l}|\rfloor$ tokens being routed to vertices in $P_{i, j}$. Moreover, the tokens routed to $P_{i, j}$ should be load-balanced.

\paragraph{Tie-Breaking The Tokens via Serialization}
We would need two tie-breaking operations here: first, for any $l$, each token with the part mark $l$ learns the portal group index $j$ (or no-op). This can be done by invoking (1) a local aggregation procedure (\Cref{lemma:local-aggregation}) that obtains the value $|T_{i, l}|$ and (2)  a local serialization procedure (\Cref{lemma:local-serialization}) that gives each token a serial number in $\{0, 1, \ldots, |T_{i, l}|-1\}$.
With the serial number and the total count, we can now assign locally the portal index $j$ for each token.

To enforce the load-balancedness requirement, the second tie-breaking must be made such that each node in $P_{i, j}$ receives roughly the same number of tokens.
This can be done by applying two additional serialization steps.
The first local serialization procedure, using the portal index $j$ as the key, assigns each token $z$ that goes to the same portal group $P_{i, j}$ a serial number $\mathit{SID}_z$.
Using the size $|P_{i, j}|$ that is preprocessed and stored at each vertex in $X_i^*$, each token $z$ obtains an index $\chi(z) := \mathit{SID}_z\bmod |P_{i, j}|$.
The second local serialization procedure can actually be preprocessed --- it assigns each portal vertex in $P_{i, j}$ a serial number.

We remark that all local aggregation and local serialization procedures described above
are actually running over the virtual graph of the corresponding part $X_i^*$.
This avoids cyclic dependency of invoking token ranking,  expander sorting, and Task 3.

Now, the problem of routing the tokens to the portals reduces to the following ``Task 2 style'' task.
In this task, each token has a portal group index $j$ and an index $\chi(z)$. The goal is to route all tokens on $X_i^*$ to the specified destination: the $\chi(z)$-th vertex within $P_{i, j}$.

\paragraph{Meet-In-The-Middle Again}
The above task can be solved using the meet-in-the-middle trick and running expander sorting (\Cref{thm:expander-sort}) twice within $X_i^*$.
In the first expander sorting, we assign for each token a key $k_z := (j, \chi(z))$ and sort the tokens within $X_i^*$.
In the second expander sorting, for each portal vertex of $P_{i, j}$ with a serial index $s$,
we create a certain number, say $\sigma_{j, s}$, of dummy tokens all with the same key $k_z := (j, s)$. The number of dummy tokens for $(j, s)$ will be exactly the same as the number of actual tokens that will be sent to this vertex.
We also add dummy tokens such that every vertex reaches the same maximum load of $L$ tokens.
These two expander sortings should now give a perfect match between the actual tokens and the dummy tokens. Thus, by reverting the routes of dummy tokens, each dummy token brings one actual token to the desired destination.

Obtaining the value $\sigma_{j, s}$ is again can be done by running a local aggregation (\Cref{lemma:local-aggregation}) over the instance where besides actual tokens, each portal vertex also creates one dummy token of the same key. After running the local aggregation, this token learns the total count $\sigma_{j, s}+1$ and goes back to the actual portal vertex.

\subsection{Merging Two Dispersed Configurations}\label{sec:merge_dispersed}
In this subsection, we will describe how to merge two dispersed configurations. First, we show that on each part $X_i^*$ and for each part mark $j$ the total number of dummy tokens is at least the number of real tokens of the same part mark.
\begin{restatable}{lemma}{dummyrealtoken}Let $T'_{i,j}$ be the $2L\cdot |X_j^*|$ dummy $T_i$-tokens who have part mark $j$ but are in part $i$ in any given dispersed configuration. Let $T_{i,j}$ be the real $T_i$ tokens who have part mark $j$ but are in part $i$ in any given dispersed configuration. For any $i,j$, we have $|T_{i,j}| \leq |T'_{i,j}|$.
\end{restatable}
\begin{proof}
\begin{align*}
|T_{i, j}| &\le 1.1 \frac{N_j}{t} + 0.1 \frac{|X|}{t^2} \tag{by \Cref{def:dispersed-configuration}}\\
&\le 1.1 \frac{L|X_j^*|}{t} + 0.1 \frac{|X|}{t^2}  \tag{condition of \textbf{Task 3}}\\
&\le 1.1 \frac{L|X_j^*|}{t} + 0.15 \frac{1}{t}\cdot \frac{|X|}{k} \tag{$t\ge (2/3)k$}\\
& \le 1.1 \frac{L|X_j^*|}{t} + 0.45 \frac{|X_j^*|}{t} \tag{$|X|/k \le 3|X_j^*|$}\\
&\le 1.75 \frac{L|X_j^*|}{t} - 0.2 \frac{|X_j^*|}{t} \tag{$L\ge 1$}\\
&\le 1.75 \frac{L|X_j^*|}{t} - 0.2 \frac{1}{t} \cdot \frac{|X|}{3k}
\tag{$|X_j^*| \ge |X|/3k$}\\
&\le 1.75 \frac{L|X_j^*|}{t} - 0.2 \frac{1}{t} \cdot \frac{|X|}{2t} \tag{$k\le (3/2)t$}\\
&\le 0.9 \frac{2L|X_j^*|}{t} - 0.1 \frac{|X|}{t^2}\ . \tag{number of dummy tokens}
\end{align*}
\end{proof}

Fix a part $X_i^*$. For clarity, we denote $\mathsf{T_{R}}$ as the set of real tokens and $\sf T_{D}$ as dummy tokens.
First of all, a local aggregation (\Cref{lemma:local-aggregation}) is invoked on $\sf T_{R}\cup T_D$ such that all dummy tokens learn $N_j := \mathit{Count}(j)-2L|X_j^*|$, the number of real tokens of part mark $j$.
Then, two local serializations (\Cref{lemma:local-serialization}) are applied to each of $\sf T_R$ and $\sf T_D$ separately.
Now, each token can set up a new key: for real token $z\in {\sf{T_R}}$ with part mark $j_z$ and local serial number $\mathit{SID}_z$, the key is $(j_z, 2\mathit{SID}_z+1)$. For a dummy token $z'\in {\sf{T_D}}$ with part mark $j_{z'}$ and local serial number $\mathit{SID}_{z'}$ the key is $(j_{z'}, 2\mathit{SID}_{z'}+2)$. We emphasize that all dummy tokens with serial number being at least $N_j$ will now be removed from $\sf T_D$ and will not participate in the final expander sorting.

Finally, an expander sort (\Cref{thm:expander-sort}) is invoked.
We can tweak the sorting algorithm to ensure that there will be an even number of tokens staying on every vertex after the sorting.
Now, each real token $z$ with the key $(j, 2\mathit{SID}_z-1)$ meets the dummy token with the key $(j, 2\mathit{SID}_z)$ at the same vertex, and \textbf{Task 3} can now be accomplished by each dummy token bringing a real token back to its starting point, which is some vertex on $X_j^*$.

After describing the recursive steps, we discuss how we solve the leaf case in \Cref{sec:task2-leaf-components} and finish the analysis on the round complexity in \Cref{sec:round_analysis}.

\subsection{Leaf Case for Task 2}\label{sec:task2-leaf-components}
In this section, we prove \Cref{lemma:leaf-components}. That is, solving \textbf{Task 2} whenever $|X|=O(n^{4\epsilon})$. On such a leaf component, it is affordable for each vertex $v\in X$ gathering the entire topology of $X$ in $O(k^8) + D(H_{X})\cdot Q(f^{0}_{H_{X}})^2 = \poly(\psi^{-1}, k, \log^{1/\epsilon}n)$ rounds during preprocessing. However, the routing task during the query time is still non-trivial.
Fortunately, we are able to apply distributed expander sorting here, and the algorithm is summarized as the following lemma.

\begin{lemma}[Leaf Components]\label{lemma:leaf-components}
Let $X\in\T$ be a leaf component.
Suppose that each vertex holds at most $L$ tokens.
Each token has a destination marker $i_z$, there are at most $L\rho_{best}$ tokens having the same destination marker.
Then, in $\poly(\psi^{-1}, k, \log^{1/\epsilon} n)$ preprocessing time and $L\cdot \poly(\psi^{-1}, \log^{1/\epsilon} n)$ routing time, each token $z$ is routed to the vertex with $i_z$-th smallest ID in $X_{best}$.
\end{lemma}
\begin{proof}%
~
\paragraph{Preprocessing}
First of all, the algorithm collects the entire virtual graph $H_X$ in $\poly(\psi^{-1}, k, \log^{1/\epsilon}n)$ time.
Once the topology of the virtual graph is obtained, in zero round the algorithm locally computes an AKS sorting network $\oldmathcal{I}_{\textsf{AKS}}$ over all vertices in $X$.
The quality of the sorting network $Q(\oldmathcal{I}_{\textsf{AKS}})$ can be set to $\poly(\psi^{-1}, \log^{1/\epsilon} n)$ as the existence is guaranteed from \cite{su2019distributed}.
We remark that there is no need to compute an all-to-best route since $X=X_{best}$ for leaf nodes.

\paragraph{Query}
Upon receiving the tokens, a meet-in-the-middle trick is applied by
setting up three sorting instances for the sorting network $\oldmathcal{I}_{\textsf{AKS}}$:
\begin{itemize}[itemsep=0pt]
    \item 
In the first pass, we invoke a local serialization (\Cref{lemma:local-serialization}) with each token's key $k_z$ being the same as the destination mark $i_z$.
\item 
In the second pass, each vertex $v\in X$ generates a dummy token with the key being its rank $i_v$ in $X$. Then, via the local aggregation (\Cref{lemma:local-aggregation}) the vertex $v$ learns the number of tokens that sets $i_v$ as its destination mark.
\item In the third and the final pass, each vertex $v\in X$ generates $\mathit{Count}(i_v)$ dummy tokens.
The $j$-th dummy token has a key being a pair $(i_v, 2 j)$.
Each query token $z$ has a destination mark $i_z$, a local serial number $\mathit{SID}_z$.
The algorithm sets each token's key to be $(i_z, 2\mathit{SID}_z-1)$.
Before running applying the precomputed sorting network $\oldmathcal{I}_{\textsf{AKS}}$, each vertex in $X$ generates up to $2 L$ extra tokens with key $\infty$ --- ensuring that each vertex in $X$ holds the same amount and has an even amount of tokens.
\item After the sorting, each token with key $(i_z, 2\mathit{SID}_z - 1)$ can be paired up with a corresponding dummy token with key $(i_v, 2 j)$. Thus, by tracing back the route of dummy token, the token $z$ can now be taken to $v$ where $i_z=i_v$, completing the task.
\end{itemize}
\paragraph{Round Complexity Analysis}
Preprocessing takes $\poly(\psi^{-1}, k, \log^{1/\epsilon} n)$.
For the query, invoking sorting tasks via the precomputed sorting network with maximum load being $2L$ takes $O(2L \log |X|) \cdot Q(\oldmathcal{I}_{\textsf{AKS}})^2 = L\cdot \poly(\psi^{-1}, \log^{1/\epsilon} n)$ rounds, as desired.
\end{proof}

\subsection{The Analysis of Round Complexity}\label{sec:round_analysis}

We are now ready to derive the desired round complexity of our routing algorithm.

\subsubsection{Deriving Recurrence Relation}

Combining the cut-matching approach and the algorithm that routes the tokens to the portals, we are now ready to analyze the round complexity.

Recall that the shuffler $\mathcal{M}_X$ contains embeddings of all virtual matchings $(M_X^1, f_{M_{X}^{1}})$, $(M_X^2, f_{M_{X}^{2}})$,$ \ldots$, $({M}_X^\lambda, f_{M_{X}^{\lambda}})$.
In each iteration $q\in \{1, 2, \ldots, \lambda\}$, the algorithm first routes the tokens to the portals.
Then, the tokens are sent along the precomputed embedded matching $f^0_{M_X^q}$.
We note that although all portals within any part $X_i^*$ are sending away roughly the same amount of tokens,
it does not mean that the number of received tokens are evenly distributed among $X_i^*$ --- they pile up at the portals.
The fact that more than $L$ tokens are piled at a vertex
affects all subsequent calls to the expander sorting procedures.
For simplicity, despite available, we do not manually apply another deterministic load-balancing algorithm such as Ghosh et al.~\cite{GhoshLMMPRRTZ99}.
The following lemma bounds the maximum load of a vertex after iteration $q$.

\begin{lemma}\label{lemma:maxload-portals}
After iteration $q$, any vertex in $X$ holds at most $L + 18qL + q^2/k = O(L \log n)$ tokens. A more detailed upper bound here is $19\lambda L \le 19\cdot (36\cdot 720\cdot 4 \cdot \ln n)\cdot L$.
\end{lemma}

\begin{proof}
Since the number of portals in $P_{i, j}^{q}$ are proportional to $m_{i,j}=m_{j, i}$, the number of tokens arrived at a portal $v\in P_{i, j}$ will be at most 
\begin{align*}
N_{\rm{max}}^{q-1} / |X_i^*| &\le  (N_{\rm{max}} + t^2q) / |X_i^*| & \text{(by \Cref{lemma:number-of-max-tokens-after-matching})}\\
&\le 18L + t^2q/|X_i^*| & \text{(by \Cref{prop:hierarchical}(\ref{itm:hierarchical:structure}))}\\
&\le 18L + q/k\ . & \text{(since $|X_i^*|\ge k^3\ge t^2k$)}
\end{align*}
The proof follows by applying an induction on $q$.
\end{proof}

\begin{lemma}\label{lemma:dispersed-configuration}
Let $X\in\T$ be a good node in the hierarchical decomposition.
Suppose that initially each vertex holds at most $L$ tokens.
Each token $z$ has a part mark $j_z$.
Assume that the shuffler $\mathcal{M}_X$ is already computed.
Then, there exists a \congest algorithm that routes all the tokens into a dispersed configuration within 
\[
O(\log n)\cdot T_{\rm sort}(6|X|/k, O(L\log n)) + O(L)\cdot (Q(\mathcal{M}_X)\cdot Q(f^0_{H_X})))^2  
\]
rounds.
\end{lemma}
\begin{proof}%
First of all, the shuffler $\mathcal{M}_X$ contains $\lambda=O(\log n)$ matchings.
At the beginning of each iteration,
the algorithm routes the tokens to the assigned portals. By \Cref{lemma:number-of-max-tokens-after-matching}, each of the destination portals will be in charge of sending at most  $(18+o(1))L$ tokens. However, by \Cref{lemma:maxload-portals}, the initial configuration allows $O(L\log n)$ tokens on any vertex.
Thus, each routing to portal procedures described in \Cref{sec:tokens_to_portal} takes $O(1)\cdot T_{\rm{sort}}(6|X|/k, O(L\log n))$ rounds, where $6|X|/k$ is the maximum possible size of a part within $X$.

In each iteration, at most $(18+o(1))L$ tokens are sent along the shuffler $\mathcal{M}_X$, a straightforward implementation by \Cref{fact:routing-along-precomputed-paths} accomplishes this step in $O(L)\cdot (Q(\mathcal{M}_X)\cdot Q(f^0_{H_X})))^2$ rounds, where $f^0_{H_X}$ is the flattened embedding of $H_X$.
Therefore, routing all tokens into some dispersed configuration can be done within
\[
O(\log n)\cdot T_{\rm{sort}}(6|X|/k, O(L\log n)) + O(L)\cdot (Q(\mathcal{M}_X)\cdot Q(f^0_{H_X})))^2
\]
rounds.
\end{proof}

\begin{theorem}\label{thm:task2-and-task3-recurrence-relation}
Let $X\in\T$ be a component in the hierarchical decomposition.
Suppose that $X$ and its all children have been precomputed with shufflers, there exist \congest algorithms for \textbf{Task 2} and \textbf{Task 3} such that:
\begin{align*}
T_2(|X|, L) &= \begin{cases}
 T_3(|X|, L) + O(L)\cdot Q(f^0_{M_{X}})^2 + T_2(6|X|/k, 4L) &\text{if $X$ is non-leaf,}\\
 L\poly(\psi^{-1}, \log^{1/\epsilon}n) &\text{if $X$ is a leaf component.}
 \end{cases}\\
T_3(|X|, L) &=O(\log n)\cdot T_{\rm sort}(6|X|/k, O(L\log n)) + O(L) \cdot (Q(\mathcal{M}_X)\cdot Q(f^0_{H_X})))^2
\end{align*}
Here $f^0_{M_{X_i}}$ is the flattened embedding from a part $X_i^*$ to the associated good node $X_i$ as defined in \Cref{prop:hierarchical}(\ref{itm:hierarchical:matching}), and $\mathcal{M}_X$ is the shuffler for the good node $X$ defined in \Cref{def:shuffler}.
\end{theorem}

\begin{proof}%
We first analyze the recurrence relation $T_3(|X|, L)$ for solving \textbf{Task 3}. The meet-in-the-middle trick involves two \Cref{lemma:dispersed-configuration}: one for the real tokens and another for dummy tokens. After the real tokens and the dummy tokens reach a dispersed configuration, the maximum load of any vertex is now $O(L\log n)$. Hence, with another
$O(T_{\rm sort}(6|X|/k, O(L\log n)))$
rounds, all real tokens meet the dummy tokens. By following the dummy tokens, \textbf{Task 3} is completed in time proportional to the round complexity of \Cref{lemma:dispersed-configuration}:
\[
T_3(|X|, L) = O(\log n)\cdot T_{\rm{sort}}(6|X|/k, O(L\log n)) + O(L)\cdot (Q(\mathcal{M}_X)\cdot Q(f^0_{H_X})))^2.
\]

We now analyze the non-leaf case of the recurrence relation $T_2(|X|, L)$ for solving \textbf{Task 2}.
First of all, a \textbf{Task 3} is involved so it takes $T_3(|X|, L)$ rounds. After solving \textbf{Task 3}, each vertex holds at most $2L$ tokens.
Recall that from \Cref{prop:hierarchical}(\ref{itm:hierarchical:matching}) not every vertex in   $X_i^*=X_i\cup X_i'$ is in the next-level good node.
The tokens on vertices in $X_i'$ have to be sent to the corresponding good node $X_i$
via the precomputed flattened embedded matching $f^0_{M_X}$, taking at most $2L\cdot Q(f^0_{M_X})$ rounds. 
After sending these tokens, each vertex on $X_i$ holds at most $4L$ tokens. Thus, in another $T_2(6|X|/k, 4L)$ rounds we can solve \textbf{Task 2} and the recurrence relation is
\[
T_2(|X|, L) = T_3(|X|, L) + O(L)\cdot Q(f^0_{M_X}) + T_2(6|X|/k, 4L)
\]
as desired. The leaf case of the recurrence relation in $T_2(|X|, L)$ directly follows from \Cref{lemma:leaf-components}, which is proved in \Cref{sec:task2-leaf-components}.
\end{proof}

\subsubsection{Analyzing Round Complexity: Solving Recurrence Relations}\label{sec:solving-recurrence-relation}

\paragraph{Query Time}
We now solve the recurrence relation and obtain the desired round compliexity for $T_2(X, L)$, $T_3(X, L)$, and $T_{\rm{sort}}(X, L)$.
By \Cref{thm:expander-sort} and \Cref{thm:task2-and-task3-recurrence-relation}, we have the following:
\begin{align*}
T_{\rm sort}(|X|, L) &= \begin{cases} 
T_3(|X|, L) + L \rho_{best}\cdot  Q(\oldmathcal{I}_{\textsf{AKS}})^2 + 
L \cdot Q(f^0_{M_X})^2 + 
T_{\rm{sort}}(6|X|/k, L) \\
\hspace*{9cm}\text{if $X$ is non-leaf,}\\
L \cdot \poly(\psi^{-1}, \log^{1/\epsilon} n) 
\hspace*{5.365cm}\text{if $X$ is a leaf component.}
\end{cases}\\
T_2(|X|, L) &= \begin{cases}
 T_3(|X|, L) + O(L)\cdot Q(f^0_{M_{X}})^2 + T_2(6|X|/k, 4L) &\text{if $X$ is non-leaf,}\\
 L\cdot \poly(\psi^{-1}, \log^{1/\epsilon}n) &\text{if $X$ is a leaf component.}
 \end{cases}\\
T_3(|X|, L) &=O(\log n)\cdot T_{\mathrm{sort}}(6|X|/k, O(L\log n)) + O(L) \cdot (Q(\mathcal{M}_X)\cdot Q(f^0_{H_X})))^2
\end{align*}

We first simplify the recurrence relation by upper bounding all non-recursing terms with $L\cdot g(\psi^{-1}, \log^{1/\epsilon} n)$ where $g(\cdot)$ is some fixed polynomial. By substituting the $T_3(|X|, L)$ term in the recurrence relation of $T_{\rm{sort}}$ and merging similar the terms, we obtain
\begin{align*}
T_{\rm{sort}}(|X|, L) \le (c_1\log n) \cdot T_{\rm{sort}}(6|X|/k, c_2L\log n) + 3L\cdot g(\psi^{-1}, \log^{1/\epsilon} n),
\end{align*}
where $c_1$ (\cite{AKS83}) and $c_2$ (\Cref{lemma:maxload-portals}) are some absolute constants. Thus, we obtain a geometric series:
\begin{align*}
T_{\rm{sort}}(|X|, L)
&\le \sum_{d=0}^{\lfloor\log_{k/6}|X|\rfloor}   (c_1c_2\log^2 n)^d \cdot 3L \cdot g(\psi^{-1}, \log^{1/\epsilon} n)\\
&\le (c_1c_2\log^2 n)^{c_3/\epsilon} \cdot 3L \cdot g(\psi^{-1}, \log^{1/\epsilon} n) \tag{$c_3 \le \epsilon (1+\log_{k/6}|X|) = 1+o(1)$}\\
&= L\cdot \poly(\psi^{-1}, \log^{1/\epsilon} n)\\
\intertext{The recursion of $T_{\rm{sort}}$ also implies that:}
T_3(|X|, L) &\le T_{\rm{sort}}(|X|, L) = L\cdot \poly(\psi^{-1}, \log^{1/\epsilon} n)\\
\intertext{
as well. Finally, by substituting $T_3(|X|, L)$ in the recursion of $T_2$ and merging the similar terms, we obtain:
}
T_2(|X|, L) &\le T_2(6|X|/k, 4L) +  (c_1c_2\log^2 n)^{c_3/\epsilon}\cdot 3L \cdot g(\psi^{-1}, \log^{1/\epsilon} n)\\
&\le 4^{c_3/\epsilon} \cdot (c_1c_2\log^2 n)^{c_3/\epsilon}\cdot 3L \cdot g(\psi^{-1}, \log^{1/\epsilon} n)\\
&= L\cdot \poly(\psi^{-1}, \log^{1/\epsilon} n) \tag{*}\label{eqT2}
\end{align*}

\paragraph{Preprocessing Time}
\begin{itemize}[itemsep=0pt]
\item 
For \textbf{Task 2}, the preprocessing relies on \textbf{Task 3} and the preprocessing for building the hierarchy. Hence we do not need any extra preprocessing steps for \textbf{Task 2} and thus: 
\begin{equation}\label{eqT2pre} T_2^{\sf{pre}}(|X|) = T_3^{\sf{pre}}(|X|). \end{equation}
\item 
For \textbf{Task 3}, the preprocessing step includes building a shuffler $\mathcal{M}_X$, which can be done in $\poly(\psi^{-1}, k, \log^{1/\epsilon} n)$ time by \Cref{lem:cut-player-potential}. The rest part of the preprocessing relies on expander sorting.
Hence, we have the following:
\begin{equation}\label{eqT3pre} 
T_3^{\sf{pre}}(|X|) = T_{\rm{sort}}^{\sf{pre}}(6|X|/k) + \poly(\psi^{-1}, k, \log^{1/\epsilon} n).
\end{equation}
\item 
For expander sorting, the preprocessing step is stated in \Cref{thm:expander-sort}:
\begin{equation}\label{eqTsortpre}
T_{\rm{sort}}^{\sf{pre}}(|X|) = \begin{cases}
T_2^{\sf{pre}}(|X|) + O(\log n)\cdot T_2(|X|, 1) + \poly(\psi^{-1}, k, \log^{1/\epsilon}n) & \text{if $X$ is non-leaf,}\\
\poly(\psi^{-1}, k, \log^{1/\epsilon}n)  & \text{if $X$ is a leaf.}
\end{cases}
\end{equation}
By applying \Cref{eqT2} with $L=1$, we obtain the recurrence relation for $T_{\rm{sort}}^{\sf{pre}}(|X|)$:
\begin{align*}
T_{\rm{sort}}^{\sf{pre}}(|X|) &= T_{\rm{sort}}^{\sf{pre}}(6|X|/k) + \poly(\psi^{-1}, k, \log^{1/\epsilon} n),
\end{align*}
which again solves to a polynomial in $\psi^{-1}$, $k$, and $\log^{1/\epsilon} n$.
\end{itemize}

Finally, we conclude the section with the following theorem statements for solving Task 2 and Task 3.

\begin{theorem}\label{thm:task2}
There exists a determinisitic \congest algorithm solving \textbf{Task 2} in preprocessing time $T_2^{\sf{pre}}(|X|) = \poly(\psi^{-1}, k, \log^{1/\epsilon}n)$ and query time $T_2(|X|, L) = L\cdot \poly(\psi^{-1}, \log^{1/\epsilon} n)$.
\end{theorem}
\begin{proof}
The statement directly follows from \eqref{eqT2pre}, \eqref{eqT3pre}, and \eqref{eqTsortpre}.
\end{proof}

\begin{theorem}\label{thm:task3}
There exists a deterministic \congest algorithm solving \textbf{Task 3} in preprocessing time $T_3^{\sf{pre}}(|X|) = \poly(\psi^{-1}, k, \log^{1/\epsilon}n)$ and
query time $T_3(|X|, L) = L\cdot \poly(\psi^{-1}, \log^{1/\epsilon} n)$.
\end{theorem}
\begin{proof}
The statement directly follows from \eqref{eqT3pre}, and \eqref{eqTsortpre}.
\end{proof}

\begin{theorem}\label{thm:expander_sorting_time}
There exists a deterministic \congest algorithm solving expander sorting in preprocessing time $T_{\rm sort}^{\sf{pre}}(|X|) = \poly(\psi^{-1}, k, \log^{1/\epsilon}n)$ and
query time $T_{\rm sort}(|X|, L) = L\cdot \poly(\psi^{-1}, \log^{1/\epsilon} n)$.
\end{theorem}
\begin{proof}
The statement directly follows from \eqref{eqTsortpre}.
\end{proof}

%% file: appendix.tex
\clearpage
\appendix
\section[The Construction of Hierarchical Decomposition from CS20]{The Construction of Hierarchical Decomposition from \cite{ChangS20}}\label{sec:hierarchical_construction}

In this section, we show how the construction of \cite{ChangS20} leads to \Cref{prop:hierarchical}, which we restate again here:%

\hierarchical*

In \cite{ChangS20}, they developed a deterministic \congest model algorithm for finding an induced subgraph of high conductance or finding a sparse cut. In particular, they can solve the following problem: 
\begin{definition}
Let $G = (V,E)$ be a graph of maximum degree $\Delta$. Let $0 < \psi_{cut} < 1$ and $0 < \phi_{emb} < 1$ be any parameters. The task
$$(\psi_{cut}, \psi_{emb}, \beta_{cut}, \beta_{leftover})\mathtt{\mhyphen{}Det\mhyphen{}Sparse\mhyphen{}Cut}$$
\end{definition}
asks for two subsets $W \subseteq V$ and $C \subseteq V$ meeting the following conditions:

\paragraph{Expander} The induced subgraph $G[W]$ has $\Psi(G[W]) \geq \psi_{emb}$.

\paragraph{Cut} The cut satisfies $0 \leq |C| \leq |V|/2$ and $\Psi(C) \leq \psi_{cut}$.

\paragraph{Balance} Either one of the following is met:
\begin{multicols}{2}
\begin{itemize}
\item $|C|\geq \beta_{cut} \cdot |V|$ and $W = \emptyset$.
\columnbreak
\item $|V \setminus (C \cup W)| \leq \beta_{leftover} \cdot |V|$.
\end{itemize}
\end{multicols}
They gave an algorithm to solve the above problem, summarized as follows.
\begin{theorem}[\cite{ChangS20}, Theorem 4.2]\label{thm:CS20algo}
Let $G = (V,E)$ be a bounded-degree graph, and let $0 < \psi_{cut} < 1$ be any parameter. The task $(\psi_{cut}, \psi_{emb}=\poly(\psi_{cut}) \cdot 2^{-O(\epsilon^{-1}\cdot \log \log n)}, \beta_{cut}=1/3, \beta_{leftover} = 1/12)\mathtt{\mhyphen{}Det\mhyphen{}Sparse\mhyphen{}Cut}$ can be solved deterministically in time:
$$\poly(D, \psi_{cut}^{-1}, \log n ) + \poly(\psi_{cut}^{-1}) \cdot 2^{O(\epsilon \log n + \epsilon^{-1} \log \log n)}$$
\end{theorem}

If we run the above algorithm with $\psi_{cut} = \Psi(G)/2$, then $C$ must be an empty set. As a result, it will find a subset of vertices $W$, with $|W| \geq |V|/12$. Moreover $G[W] \geq \psi_{emb}$. A closer look into the construction of their algorithm shows that to certify the conductance of $W$, it actually returns an embedding $H_{W}$ on $G[W]$ with congestion and dilation of at most $\psi_{emb}$.

Moreover, due to the recursive nature, it actually produces a hierarchy of embeddings. We summarize their construction for $(\psi_{cut}, \psi_{emb}=\poly(\psi_{cut}) \cdot 2^{-O(\epsilon^{-1}\cdot \log \log n)}, \beta_{cut}=1/3, \beta_{\textit{leftover}} = 1/12)\mathtt{\mhyphen{}Det\mhyphen{}Sparse\mhyphen{}Cut}$. First they showed the above can be reduced to solving $O(1)$ instances of a weaker problem with $\beta_{cut}=1/100$, $(\psi_{cut}, \psi_{emb}=\poly(\psi_{cut}) \cdot 2^{-O(\epsilon^{-1}\cdot \log \log n)}, \beta_{cut}=1/100, \beta_{\textit{leftover}} = 1/12)\mathtt{\mhyphen{}Det\mhyphen{}Sparse\mhyphen{}Cut}$. If the algorithm returns an embedding in the end, then exactly one of the $O(1)$ instances return an embedding. We describe how this  weaker problem can solved recursively as follows:

\begin{enumerate}[leftmargin=*]
\item Let $(V,E)$ be the current base graph. Partition $V$ into $V_1 \ldots V_k$ where each $|V_i| \in \{\lfloor |V|/k \rfloor, \lceil |V|/k \rceil\}$ and $\max_{x\in V_i} \ID(x) \leq \min_{y \in V_{i+1}} \ID(y)$ for $1 \leq i < k$.  
\item For each $V_i$ try to embed an expander $H_i$ onto it. Each $V_i$ may be success (an expander has been embedded), fail (a sparse cut has been bound), or active at any given round.  $H_i$ is empty in the beginning. For each round $r=1 \ldots \lambda = O(\log n)$, for each $1\leq i \leq k$, if $V_i$ is still active, we try to extend $H_{i}$ with one matching by performing a round of cut-matching game. 
\begin{enumerate}[leftmargin=*]
    \item {\bf Cut Player:} Try to find a sparse cut of $H_i$ by solving $(\widehat{\psi_{cut}} = 1/2, \widehat{\psi_{emb}}=\poly(\psi_{cut}^{-1}, \log n)\cdot \psi_{emb}, \widehat{\beta_{cut}} = 1/3, \widehat{\beta_{\textit{leftover}}}=1/12)\mathtt{\mhyphen{}Det\mhyphen{}Sparse\mhyphen{}Cut}$ on the current $H_{i}$. 

    If a cut $C_i$ with $|C_i| \geq |V_i|/4$ has been returned then proceed to the matching player. Otherwise, an expander embedding on $U_i \subseteq V_i$ must have been returned with $|U_i| \geq (2/3)\cdot |V_i|$, since $|V_i| - |C_i \cup U_i| \geq |V_i|/12$, which implies $|U_i| \geq 11|V_i|/12 - |V_i|/4 \geq (2/3)\cdot |U_i|$. In this case, the status of $V_i$ becomes success.

    \item {\bf Matching Player:} Recall that $C_i$ is the cut found by the cut player. For each $V_i$ that is still active, simultaneously each $V_i$ try to embed a matching $M_i$ with congestion $c= \tilde{O}(\psi_{cut}^{-2})$ and dilation $d = \tilde{O}(\psi^{-1}_{cut})$ from $C_i$ to $V_i \setminus C_i$. This is done by parallel DFS algorithm in \cite[Theorem D.11]{ChangS20}.  For each $V_i$, either a large fraction of $C_i$ is saturated, in which we will successfully add one matching $M_i$ to $H_i$. Then, it will continue into the next iteration. Otherwise, we will find a cut that contains at least $|V_i|/8$ vertices and then set the status of $V_i$ to be fail. There is a possibility the matching player finds a global sparse cut $C$ with $\Psi(C)\leq \psi_{cut}$ and $|C|\geq |V|/8$ and then terminates directly.
\end{enumerate}
It was shown in \cite{KKOV07} that after $O(\log n)$ iterations such a cut-matching game will terminate, which means each $V_i$ is either success or fail. If a large fraction of $V_i$ fails, then the algorithm will return a cut. Otherwise, the algorithm tries to merge different expanders $U_i \subseteq V_i$ from those who succeeded.

\item\label{itm:merge} In the following we describe how such an expander is constructed. The algorithm may detect if the graph has a sparse cut during the construction, on which it will return the cut and exit. However, we will focus on the case when an expander is returned, since this is when it will be added to the hierarchical decomposition. Here, we do not attempt to show the correctness as it has been proven in \cite{ChangS20}, but rather we focus on describing how it is constructed and how they form the hierarchical decomposition. 

W.l.o.g.~let $V_1 \ldots V_{k'}$ be the parts that have succeeded and let $U_1 \ldots U_{k'}$ be the vertices where the expanders $H_1 \ldots H_{k'}$ that have been embedded. Let $T = U_1 \ldots U_{k'}$. They show that if the algorithm did not return a cut (the case where a large fraction of parts succeeded), then $|T| \geq (11/20)\dot |V|$. Let $S = V \setminus T$. Now run the parallel DFS algorithm of \cite[Theorem D.11]{ChangS20} again to find a set of paths $\mathcal{P}$ of congestion and dilation $\polylog(n)\cdot O(\psi_{cut})$ that connect between $S$ and $T$. The algorithm will match up most of the vertices in $S$ except a few of them.  

Let $U'_{i} \subseteq S$ be the set of vertices that are matched to $U_i$ via paths in $\mathcal{P}$. Let $U^{*}_i = U_i \cup U'_i$. We also create $H^{*}_i$ with vertex set $V(H^{*}_i) = V(H_i) \cup V(U'_i)$ and edge set $E(H^{*}_i) = E(H_i) \cup \{(x,y)\mid x \in U_i, y \in U'_i, \mbox{$xy$ connected by some path in $\mathcal{P}$} \}$.

A crucial observation in their construction is that if the algorithm returns an expander, then it will be the union of some $U^{*}_i$'s, say w.l.o.g.(after reordering)~we assume they are $U^{*}_1 \ldots U^{*}_{k''}$. Consider $G' = G[U^{*}_1 \ldots U^{*}_{k''}]$. A property that they have shown is that for all cuts $S$ of $G'$ that respect $\{U^{*}_1 \ldots U^{*}_{k''} \}$, $\Phi(S) \geq \psi_{cut}/8$ (\cite[Lemma 4.7]{ChangS20}). By combining this property with the fact that $\Psi(H^{*}_i) = \Omega(\Psi(H_i)) = \Omega(\poly(\psi_{cut}^{-1}, \log n)\cdot \psi_{emb})$ and that $H^{*}_i$'s are simultaneous embedding into $G$ with quality $\tilde{O}(\psi^{-1}_{cut})$, they showed that $\Psi(G') \geq \psi_{emb}$. That is, they showed that the union of these vertices $U^{*}_1 \ldots U^{*}_{k''}$ have slightly weaker expansion than individual $H^{*}_i$'s. However, since each $H_i$ was embedded with slightly higher expansion (i.e.~$\poly(\psi_{cut}^{-1}, \log n)\cdot \psi_{emb}$), the loss on expansion in combining them still satisfies our purpose (i.e.~the expansion is still lower bounded by $\psi_{emb}$).

Now we describe its connection to the hierarchical decomposition. If the algorithm returns an expander $G[U^{*}_1 \ldots U^{*}_{k''}]$ and $|U^{*}_1 \ldots U^{*}_{k''}| \geq (2/3)\cdot |V|$, we will create a node $X$ and add it to $\T$, where we set the $X = U^{*}_1 \ldots U^{*}_{k''}$. Moreover, we set $t = k''$ and then we set the good children of $X$, $X_1 \ldots X_t$ to be that $X_1 = U_1 \ldots X_t = U_k$. Note that since this recursion is proccessed in the post-order, the children $X_1 \ldots X_t$ must have been already added to $\T$. Then we create the bad children of $X$, $X'_1 \ldots X'_t$, to be that $X'_1 = U'_1 \ldots X'_t = U'_{k''}$. 
\end{enumerate}

First we check \Cref{prop:hierarchical}(\ref{itm:hierarchical:structure}). Suppose that $X$ is a good internal node. The fact that $$\max_{x \in X_{i}} \ID(x) \leq \min_{y \in X_{i+1}} \ID(y)$$ holds because of how $V_i$'s are partitioned and each $X_i \subseteq V_i$. Moreover, for each $X^{*}_i$ we have:
\begin{align*}
|X^{*}_i| &\leq |X_i| + |X'_i| \\
&\leq 2|X_i| && |X'_i| \leq |X_i|\\
&\leq 2|V_i| && X_i \subseteq V_i\\
&\leq 4|V|/k && |V_i| \leq 2|V|/k\\
&\leq 6|X|/k && |X| \geq (2/3) |V|
\end{align*}
For the lower bound side, we have:
\begin{align*}
|X^{*}_i| &\geq |X_i| \\
&\geq (2/3)\cdot |V_i|  && \mbox{$X_i$ is successful embedding to $V_i$}\\
&\geq (1/3)\cdot |V|/k && |V_i| \geq (1/2)\cdot(|V|/k)\\
&\geq (1/3)\cdot |X|/k 
\end{align*}

For the second set of inequalities, observe that $|V|/k - 1 \leq |V_i| \leq |V|/k + 1$ and we set $\tau = |V|/k$. 

To verify \Cref{prop:hierarchical}(\ref{itm:hierarchical:embedding}), we observe how the parameters $\psi_{cut}$ and $\psi_{emb}$ in different levels of the recursion. Suppose that the root level has level 0. Let $\ell$ be the current level, we have: \begin{multicols}{2} 
\centering
$\psi_{cut} = \begin{cases}1/3 & \ell \geq 1 \\ \psi /2  & \ell = 0 \end{cases} $
\vspace*{\fill}
\columnbreak
$\psi_{emb} = \begin{cases} 1/(\log^{O(\ell(T) - \ell)} n)  & \ell \geq 1 \\ \poly(\psi)/ \log^{O(1/\epsilon)} n  & \ell = 0 \end{cases} $
\end{multicols}

If $X$ is a good internal node, then the simultaneous embedding of $X_1 \ldots X_t$, $\bigcup_{i=1}^{t} H_{X_{i}}$ has congestion and dilation $\polylog(n) \cdot O(\psi_{cut}^{-1})$ in $X$. Moreover, for any good node $X$, the recursion must have found an expander here. Thus, we have $\Psi(H_{X}) \geq \psi_{emb}$.

\Cref{prop:hierarchical}(\ref{itm:hierarchical:matching}) follows by the construction in Step \ref{itm:merge} of the algorithm above, where $X$ is formed by the union of $U^{*}_i \ldots U^{*}_{k''}$, where $U^{*}_i = U_i \cup U'_i$. Note that we have set $X_i = U_i$ and $X'_i = U'_i$. The set of paths $\oldmathcal{P}$ are paths in $H_{X}$ that connects between $U_1 \cup \ldots \cup U_{k''}$ and $U'_1 \cup \ldots \cup U'_{k''}$ with congestion and dilation $\polylog(n) \cdot O(\psi^{-1}_{cut})$. Moreover, each vertex in $U'_i$ is connected to some vertex in $U_i$ via a path in $\oldmathcal{P}$ and each vertex in $U_i$ is the endpoint of at most one path in $\oldmathcal{P}$. Therefore, $|U'_i| \leq |U_i|$ and so $|X'_i| \leq |X_i|$.

\section{The Construction of Shufflers}\label{sec:the-modified-cut-matching-game}

In this section we prove~\Cref{lem:disperse} by giving implementation details for the cut player and the matching player.

\subsection{Cut Player}
At the $i$-th iteration, the cut player computes two disjoint vertex subsets $S, S'\subseteq V(Y)$ with the following properties:

\begin{property}\label{prop:cut}~
\begin{enumerate}[itemsep=0pt]

\item \label{itm:size} $|S_{X}| < |S'_{X}|$.
\item \label{itm:potential} Consider any mapping $\sigma: S\to S'$. Then,
$$\sum_{y\in S} \left\| R_{i-1}[y] - R_{i-1}[\sigma(y)]\right\|^2 \geq \frac{1}{720} \cdot \Pi(i-1)\ . $$
\end{enumerate}
\end{property}
\begin{restatable}{lemma}{cutplayerpotential}
\label{lem:cut-player-potential} There exist subsets $S$ and $S'$ with \Cref{prop:cut}. Moreover, they can be computed in $\poly(\psi^{-1}, k, \log^{1/\epsilon} n)$ rounds.
\end{restatable}

The following \Cref{lem:Gaussian} and \Cref{lem:RSTcut} are very useful for the proof.

\begin{lemma}[{\cite[Lemma 3.5]{KRV09}}]\label{lem:Gaussian}Let $v\in \mathbb{R}^{t}$ be a vector with $\lVert v \rVert = \ell$, and let $r \in \mathbb{R}^{t}$ be a uniformly random vector. We have 
$\E[ \lVert v\cdot r \rVert^2 ] = \ell^2/t$.
\end{lemma}

\begin{lemma}[{Properties of $A^{l}$ and $A^{r}$ \cite[Lemma 3.3]{RST14}}]\label{lem:RSTcut} 
Let $A$ be any set.
Consider any mapping $\mu: A \to \mathbb{R}$ and define $\bar{\mu} = (\sum_{x \in A} \mu(x))/|A|$.
There exists two disjoint sets $A^{l}, A^{r} \subseteq A$ and a separation value $\gamma$ such that the following holds:
\begin{enumerate}
\item \label{itm:RSTcut1} Either $\max_{v \in A^{l}} \mu(v) \leq \gamma \leq \min_{v \in A^{r}} \mu(v)$ or $\max_{v \in A^{r}} \mu(v) \leq \gamma \leq \min_{v\in A^{l}} \mu(v)$.
\item  \label{itm:RSTcut2} For each $v \in A^{l}$, $|\mu(v) - \gamma| \geq |\mu(v) - \bar{\mu}|/3$.
\item \label{itm:RSTcut3}  $|A^{l}| \leq |A|/8$ and $|A^{r}| \geq |A| /2$.
\item \label{itm:RSTcut4}  $\sum_{v \in A^{l}}|\mu(v) - \bar{\mu}|^2 \geq (1/80)\cdot \sum_{v \in A} |\mu(v) - \bar{\mu}|^2$.
\end{enumerate}
\end{lemma}

\begin{proof}[Proof of \Cref{lem:cut-player-potential}]

Let $r\in\mathbb{R}^t$ be a random unit $t$-dimensional vector orthogonal to $\mathbf{1}$. Let $\mu = R_{i-1} \cdot r$. For each $y\in V(Y)$, $\mu[y]$ is a mapping of $y$ to a real number.
By \Cref{lem:RSTcut}, with the same $\mu$ mapping and $A = V(Y)$.
there exists disjoint sets $A^{l}, A^{r}$ and value $\gamma$ with the properties listed in the lemma statement of \Cref{lem:RSTcut}.
Now we set $S = A^{l}$ and $S' = A^{r}$. 

It is now straightforward to check that \Cref{prop:cut}(\ref{itm:size}) holds:
by \Cref{prop:hierarchical}(\ref{itm:hierarchical:structure}), there exists $\tau = \Theta(|X|/k)$ such that for any $y \in Y$, we have $(2/3)(\tau - 1) \leq |\{y\}_{X}| \leq 2(\tau +1)$.
Therefore, with \Cref{lem:RSTcut}(\ref{itm:RSTcut3}), we have:
\begin{align*}
|S_{X}| &\leq  2(\tau + 1)|Y|/8 
= (\tau + 1)t/4 = (\tau t)/4 + t/4 
\leq (1+1/n^{\epsilon})(\tau t)/4 \text{, and }\\
|S'_{X}| &\geq (2/3)(\tau - 1) \cdot (|Y|/2) 
\geq (1/3)\tau t - (1/3) t 
\geq (1-1/n^{\epsilon})(\tau t)/3 .
    \end{align*}
Thus, with $\tau \ge n^{2\epsilon}$, we have $|S_{X}| < |\bar{S}_{X}|$.

Now, given $\sigma(y)$ for every $y \in S$, we have:
\begin{align*}
\sum_{y\in S} \| R_{i-1}[y] - R_{i-1}[\sigma(y)]\|^2 &= t\cdot \sum_{y \in S} \E[(\mu[y] - \mu[\sigma(y)])^2] && \text{(by \Cref{lem:Gaussian})}\\
&\geq t\cdot \sum_{y \in S} \E[(\mu[y] - \gamma)^2] && 
\text{(by \Cref{lem:RSTcut}(\ref{itm:RSTcut1}))} \\
&\geq t \cdot \sum_{y \in S} \E[(\mu[y] - \bar{\mu})^2]/9   &&  \text{(by \Cref{lem:RSTcut}(\ref{itm:RSTcut2}))}\\
&\geq t \cdot \sum_{y \in V(Y)} \E[(\mu[y] - \bar{\mu})^2]/720  &&  \text{(by \Cref{lem:RSTcut}(\ref{itm:RSTcut4}))}\\
&= (1/720)\cdot \sum_{y \in V(Y)} \left\lVert R_{i-1}[y] - \bm{\frac{1}{|Y|}} \right\rVert^2  && \text{(by \Cref{lem:Gaussian})}\\\
&= (1/720) \cdot \Pi(i-1)
\end{align*}

Since subsets of such property exists. We can have each node in $X$ to learn the topology of $Y$ in $\poly(k) \cdot \log^{O(1/\epsilon)} n$ rounds. Every node then compute such subsets locally by enumerating all possible disjoint sets $S$ and $S'$ and agree on the lexicographical smallest ones that satisfy the property. 
\end{proof}

\Cref{lem:cut-player-potential} shows that if the current random walk induced by $(M^{1} \ldots M^{i-1})$ is not mixing then there exist two disjoint subsets $S$ and $S'$ such that the probability of the random walk ending at nodes in $S$ and nodes in $S'$ are quite different. In particular, \Cref{prop:cut}(\ref{itm:potential}) says if we take $|S|$ pairs of vertices (one per vertex in $S$) between $S$ and $S'$, the sum of the distance on the distribution of the pairs takes up to a constant fraction of the current potential. Thus, if we add a matching between $S$ and $S'$ in the next iteration, the potential will drop significantly. The proof of \Cref{lem:cut-player-potential} (which we defer the appendix) follows from an adaption of the framework of \cite{KRV09, RST14}. Also, since every node in $X$ can learn the cluster graph $Y$ efficiently, the sets $S$ and $S'$ can be computed locally at each vertex.  Now we describe our implementation to the matching player.

\subsection{Matching Player}

At the $i$'th iteration, suppose that $(S, S')$ are the disjoint subsets computed by the cut player. Let $(S_{X}, S'_X)$ be the corresponding subsets in $X$. The matching player finds a virtual matching $M^{i}_{X}$ that covers every vertex in $S_{X}$ and an embedding $f_{M^{i}_{X}}$ of $M^{i}_{X}$ onto $X$ with good quality.

The matching player invokes 
\Cref{lem:deterministic_path_embedding}
with $\psi = \Psi(H_{X})/2$, $S = S_{X}$, and $T = S'_{X}$. 
Once the matching $M^{i}_{X}$ and the embedding $f_{M^{i}_{X}}$ are found, we set $M^{i}$ be the corresponding natural fraction matching in $Y$.
Note that if an edge $(u,v)$ has value $x_{uv}$ in $M^{i}$, then it corresponds to exactly $n' \cdot x_{uv}$ different paths  in $f_{M^{i}_{X}}(M^{i}_{X})$ that connects between $X^{*}_{u}$ and $X^{*}_{v}$.

\subsection{Termination} The algorithm terminates when $\Pi(i) \leq 1/(9n^3)$. Note that the potential function can be computed in $\poly(\psi^{-1}, k, \log^{1/\epsilon} n)$ rounds after each iteration.
The following lemma concludes that the number of iterations will be at most $O(\log n)$.

\begin{lemma}\label{lem:size-of-shuffler}\label{lem:cut-matching-iterations}
The algorithm terminates in $O(\log n)$ iterations.
\end{lemma}

\begin{proof}
Consider iteration $i$, we show that after we construct the fractional matching $M^i = \{m_{ab}\}_{(a,b) \in \binom{V(Y)}{2}}$ in iteration $i$ the potential decrease is:
\begin{equation}\label{eq:potential-difference}
\begin{aligned}
\Pi(i-1) - \Pi(i) &= \sum_{y \in V(Y)} \left\| R_{i-1}[y] - \bm{\frac{1}{|Y|}}\right\|^2
-  \sum_{y \in V(Y)} \left\|R_{i}[y] - \bm{\frac{1}{|Y|}}\right\|^2 \\
&=  \sum_{y \in V(Y)} \left( \| R_{i-1}[y] \|^2 -  \|R_{i}[y]\|^2\right) \hspace*{1.6cm}\mbox{(since $R_i\cdot \bm{\frac{1}{|Y|}} = R_{i-1}\cdot \bm{\frac{1}{|Y|}} = \frac{1}{|Y|}$)}
\end{aligned}
\end{equation}

Here we apply a similar argument from \cite[Lemma 3.4]{RST14}.
However, our definition of $R_{i}[y]$ is completely different from theirs so the derivation is \emph{slightly} different.
Let $M^{i*} =\{m_{ab}^*\}$ be the perfect fractional matching where self-loops are added to $M^i$ such that $\sum_b m_{ab}^*=1$.
We now observe that for all $a\in V(Y)$,
$R_i[a] = (\frac12R_{i-1}[a] + \frac12 \sum_b m_{ab}^* R_{i-1}[b])$.
Thus,
\begin{align*}
\|R_{i-1}[a]\|^2 - \|R_{i}[a]\|^2 &=
\|R_{i-1}[a]\|^2 - \left\|\frac12 R_{i-1}[a] + \frac12 \sum_b m^*_{ab} R_{i-1}[b]\right\|^2\\
&=\frac14\left(\|R_{i-1}[a]\|^2  - \left\|\sum_b m^*_{ab} R_{i-1}[b]\right\|^2\right)
+ \frac12 \left\langle R_{i-1}[a], R_{i-1}[a] - \sum_b m^*_{ab}R_{i-1}[b] \right\rangle
\intertext{We now use Jensen's inequality
and observe that $\|\sum_b m^*_{ab} R_{i-1}[b]\|^2 \le \sum_b m^*_{ab} \| R_{i-1}[b]\|^2$. This gives
}
&\ge 
\frac14\left(\|R_{i-1}[a]\|^2  - \sum_b m^*_{ab}\left\| R_{i-1}[b]\right\|^2\right)
+ \frac12 \left\langle R_{i-1}[a], R_{i-1}[a] - \sum_b m^*_{ab}R_{i-1}[b] \right\rangle
\end{align*}
By summing up the above inequality with all $a\in V(Y)$, with the fact that $m_{ab}^* = m_{ba}^*$ and $\sum_{b} m_{ab}^* = 1$, we have:
\begin{align*}
\sum_{a\in V(Y)} \frac14\left(\|R_{i-1}[a]\|^2 - \sum_b m_{ab}^* \|R_{i-1}[b]\|^2\right) = 0\ .
\end{align*}
Now, plugging the above observation into \eqref{eq:potential-difference}, we have:
\begin{align*}
\Pi(i-1) - \Pi(i) &\ge \sum_{a\in V(Y)} \frac12 \left\langle R_{i-1}[a], R_{i-1}[a] - \sum_b m^*_{ab}R_{i-1}[b] \right\rangle \\
&= \sum_{a\in V(Y)}\sum_{b\in V(Y)} \frac{m_{ab}^*}{2} \left\langle R_{i-1}[a], R_{i-1}[a] - R_{i-1}[b]\right\rangle  \tag{$\sum_b m_{ab}^*=1$}\\
\intertext{By our choice of the fractional matching $M$, we know that if $a\neq b$ and $m_{ab}^*\neq 0$, then $a$ and $b$ belong to different side of $S$ and $S'$. Furthermore, $m_{ab}^*=m_{ab}$ whenever $a\neq b$. Hence, the above expression can be further simplified to:}
&= \sum_{a\in S}\sum_{b\in S'} \frac{m_{ab}}{2} \left\langle R_{i-1}[a], R_{i-1}[a] - R_{i-1}[b]\right\rangle + 
\sum_{a\in S'}\sum_{b\in S} \frac{m_{ab}}{2} \left\langle R_{i-1}[a], R_{i-1}[a] - R_{i-1}[b]\right\rangle\\
\intertext{Now, we use a trick of renaming the second part of the summation. Using $m_{ab}=m_{ba}$ and merge the corresponding terms carefully:}
&= \sum_{a\in S}\sum_{b\in S'} \frac{m_{ab}}{2} \left\langle R_{i-1}[a], R_{i-1}[a] - R_{i-1}[b]\right\rangle + 
\sum_{b\in S'}\sum_{a\in S} \frac{m_{ba}}{2} \left\langle R_{i-1}[b], R_{i-1}[b] - R_{i-1}[a]\right\rangle\\
&=\sum_{a\in S}\sum_{b\in S'} \frac{m_{ab}}{2} \left\langle R_{i-1}[a] - R_{i-1}[b], R_{i-1}[a] - R_{i-1}[b]\right\rangle\\
&=\sum_{a\in S}\sum_{b\in S'} \frac{m_{ab}}{2} \|R_{i-1}[a]-R_{i-1}[b]\|^2\\
\intertext{Recall that $m_{ab}$ can be viewed as the number of virtual matchings from the $a$-th part to the $b$-th part of a underlying good node $X$.
By the implementation of the matching player, we know that there is a matching from $S_X$ to $S'_X$ saturating $S_X$.
Therefore, since each part of $X$ has at least $n'/18$ vertices, by partitioning the virtual matching, we obtain at least $n'/18$ matchings  from $S$ to $S'$.
Let $\{\sigma_j\}$ be such mappings from $S$ to $S'$. The above expression can then be lower bounded by:}
&\ge \frac{1}{2n'}\sum_{j=1}^{n'/18}\sum_{a\in S} \|R_{i-1}[a] - R_{i-1}[\sigma_j(a)]\|^2\\
&\ge \frac{1}{2n'} \cdot \frac{n'}{18} \cdot \frac{1}{720}\cdot \Pi(i-1)\ . \tag{by \Cref{lem:cut-player-potential}}
\end{align*}

Therefore, $$\Pi(i) \leq \left(1-\frac{1}{36\cdot 720} \right) \cdot \Pi(i-1) \ .$$
At the beginning of the cut-matching game, we have $\Pi(0)= t-1$.
By setting $\lambda = 36\cdot 720\cdot 4\cdot \ln n$, we know that $\Pi(\lambda) \le 1/(9n^3)$.
This implies that the algorithm terminates in $\lambda = O(\log n)$ rounds.
\end{proof}

\section{Distributed Expander Sorting}\label{sec:distributed-expander-sorting-appendix}

\deterministicexpandersort*

\begin{proof}[Proof of \Cref{thm:expander-sort}]
Without loss of generality, we may assume that $L=1$, and that every vertex holds exactly one token.
The idea is to embed an entire AKS sorting network \cite{AKS83} over the vertices of $X_{best}$.
Once we have such an embedding, the algorithm can then be done in three
steps:
\begin{description}[itemsep=0pt]
\item[Step 1.] Send all tokens to $X_{best}$ in any way but load-balanced: each vertex in $X_{best}$ holds at most $L\cdot\rho_{best}$ tokens.
\item[Step 2.] Simulate the AKS sorting network to sort all the tokens on $X_{best}$.
\item[Step 3.] Using a precomputed order-preserving all-to-best route that sends the tokens back to each vertex, preserving the sortedness.
\end{description}

It is straightforward to see that, the chosen route in the first step can be the same as the third step.
Therefore, it suffices to build an order-preserving all-to-best route.
We note that if
$X$ is a leaf component, then there is no need to perform Step 1 and Step 3 as $X=X_{best}$.

Now,  suppose that $X$ is a non-leaf component.
We notice that the third step is not trivial as the IDs of vertices in $X$ may not be consecutive, so we cannot directly invoke a single \textbf{Task 2}.
Fortunately, in the preprocessing step it is affordable to perform $k$ binary search procedures, so each vertex $v\in X^*$ knows which part $X_{j_v}^*$ it is mapped to.

To compute an order-preserving all-to-best route, the algorithm 
creates a token at each vertex $v\in X^*$ with the part mark $j_v$, and then invoke \textbf{Task 3} to route these tokens to the associated part.
Then, for each $j$, the algorithm invokes the expander sorting (\Cref{thm:expander-sort}) within the part $X_j^*$.
After the tokens are sorted within $X_j^*$, an order-preserving all-to-best route can be constructed by concatenating the all-to-best routes in the parts.

An AKS sorting network $\oldmathcal{I}_{\textsf{AKS}}$ has $O(\log |X_{best}|)$ layers, and in each layer there is a matching specifying which pairs of the keys are being compared.
During the preprocessing time, each vertex in $X_{best}$ knows its rank, so an actual ``routable'' sorting network can be constructed by invoking \textbf{Task 2} $O(\log n)$ times. 
Since there are $\rho_{best}$ tokens on each best vertex, in every comparison two sets of $\rho_{best}$  tokens are first merged and then split into two halves.
When $L>1$ the same argument applies since in each comparison two sets of $L\rho_{best}$ tokens are being compared.

The theorem statement concerns $X^*=X\cup X'$, where additional $2Q(f^0_M)^2$ rounds are added for moving tokens back and forth between $X'$ and $X$.
\end{proof}

\tokenranking*

Here we provide a reduction to expander sorting.
In the preparation step, the algorithm obtains the rank of each node.
To achieve this, the algorithm first builds an arbitrary serialization that assigns each vertex a serial number in $\{1, 2, \ldots, n\}$. This can be done by obtaining a BFS tree and assigning each vertex the in-order index in the diameter time $D(H_X)\cdot Q(f^0_{H_X})^2$.
Then, each vertex generates a token with this serial number as the key, setting $L=1$ and invoking \Cref{thm:expander-sort}.
After the sort, the vertex $v$ obtains a token with key equals $\mathit{rank}(v)$, which is the rank of $\mathit{ID}(v)$ among all vertex IDs.
Now we solve the real token ranking problem.

An ideal case is that all keys are distinct.
In this case,
each vertex generates extra tokens with $k_z=\infty$ such that each vertex holds \emph{exactly} $L$ tokens.
After invoking the expander sorting procedure (\Cref{thm:expander-sort}), we know that each vertex with rank $i=\mathit{rank}(v)$ now holds exactly $L$ tokens where those token ranks are in the range $[(i-1)L, iL)$.
Thus, all token ranks can be assigned correctly.

Now we solve the general case.
The idea is to perform deduplication.
One possible implementation is using expander sort described in \Cref{thm:expander-sort} with a tiny add-on:
Initially, each token $z$ is tagged with its starting location vertex ID and its sequential order among all tokens within that starting location, denoted as $u_z$.
This tag $u_z$ is used for tie-breaking.
When running the expander sorting algorithm, whenever two tokens with identical keys are being compared, the token with a larger tag is \emph{marked as a duplicate} and considered as a larger value. 
The following property implies that for any \emph{comparison-based} expander sorting algorithm, for each key there will be exactly one token (the token with the smallest tag)
not marked as a duplicate.
\begin{property}[Chain of Comparisons]\label{sorting-chain}
Let $\oldmathcal{A}_{\textsf{sort}}$ be any deterministic comparison-based sorting algorithm that implements an expander sort.
Let $Z$ be a set of tokens with the same key.
Let $z^*$ be the token in $Z$ with the smallest tag $u_{z^*}$.
For any other token $z\in Z$, $z\neq z^*$, 
there exists a chain of tokens $(z_0=z, z_1, z_2, \ldots, z_\ell=z^*)$ in $Z$ such that (1) $u_{z_i} > u_{z_{i+1}}$ for all $i$, and (2) the comparisons $(z_i, z_{i+1})$ take place in chronological order during the execution of $\oldmathcal{A}_{\textsf{sort}}$.
\end{property}

\begin{proof}
To see this, we observe that such a sequence can be obtained by first slightly decreasing $z$'s key $k_z$ and running $\oldmathcal{A}_{\textsf{sort}}$, and recording all tokens in $Z$ that got compared with $z$ but with a tag smaller than the currently observed smallest tag.
If $z\neq z^*$, then this sequence must end at $z_{\rm{end}}=z^*$, as otherwise the sorting algorithm may not be sorting properly: $\oldmathcal{A}_{\textsf{sort}}$ results the same on the instance where we swap the tags between $z_{\rm{end}}$ and $z^*$.
\end{proof}

The correct token ranks can then be obtained by running \Cref{thm:expander-sort} again without the participation of the duplicated tokens.
However, we still need to propagate these ranks to all tokens that were marked as duplicates. The following lemma shows that propagation can be done by reverting the expander sort.

\localpropagation*
\begin{proof}
To achieve local propagation, one may simply revert the entire expander sorting procedure, and propagate the information whenever needed (i.e., upon comparing $z$ and $z'$, update \emph{both} tag and variable if the key is the same $k_{z'}=k_z$ but the tag is larger $u_{z'} > u_z$).
The correctness follows again from the chain of comparison property described above in \Cref{sorting-chain}.
\end{proof}

\begin{proof}[Proof of \Cref{thm:token-ranking}]
\Cref{thm:token-ranking} follows after the above discussion and \Cref{lemma:local-propagation}.
\end{proof}

\localserialization*

\begin{proof}
This can be done by invoking token ranking twice.
First, the algorithm invokes \Cref{thm:token-ranking} with the keys being attached with a unique tag (starting vertex ID and its serial number).
Then, the algorithm invokes the token ranking again but when applying the propagation all tokens with the same key obtain the smallest rank among them.
Finally, the local rank can be obtained by subtracting the global rank with the propagated value.
\end{proof}

\localaggregation*

\begin{proof}
This can be done by invoking two token rankings (\Cref{thm:token-ranking}) where the second time all tag values are negated.
After that, the token with the smallest tag value obtains the count.
With the count value, by applying a
local propagation (\Cref{lemma:local-propagation}), all tokens obtain the correct count values.
\end{proof}

%% file: Task123.tex
\section{The Reduction from Task 1 to Task 2}\label{sec:reductions}
To reduce from {\bf Task 1} to {\bf Task 2}, we will delegate the routing task to the vertices in $V_{best}$.
The intuition leads to a definition of a load-balanced \emph{all-to-best} mapping $h:V\to V_{best}$ with a bounded pre-image size $\rho_{best}$. With the mapping $h$, the routing task can be accomplished by first sending tokens with destination $v$ to $h(v)\in V_{best}$ and then routing the token from $h(v)$ to $v$.
The latter half of the above solution can also be thought of first solving a routing task from $v$ to $h(v)$ and then reversing the path. However, the computation of $h(v)$ is not always available locally on any vertex. It depends on the design of $h$. If all vertices do not agree on the same $h(v)$ value, the above reduction will not work.

\paragraph{An Ideal Case}
In an ideal scenario where all vertices have IDs exactly from $\{1, 2, \ldots, n\}$, such a function $h$ can be easily defined by setting $h(v)$ to be the $i$-th smallest ID among $V_{best}$ where $i=\mathrm{ID}(v)\bmod |V_{best}|$.
This definition of $h$ has two benefits: On one hand, computing the size of $|V_{best}|$ can be done in a straightforward way (e.g., \cite[Lemma A.1]{ChangS20}) such that $i$ can be computed locally.
On the other hand, the indirect value ``$i$'' on a token can be rewritten recursively along the hierarchy $\T$.
The routing task on a good node $X\in\T$ can then be reduced to sending tokens to the corresponding next-level components within $X$.

The following \textbf{Task 1'} summarizes the ideal case.

\paragraph{Task 1'}
Let $G$ be a constant degree $\psi$-expander, where each vertex of $G$ has a unique destination ID in $\{0, 1, \ldots, n-1\}$. Suppose that each node in $G$ holds at most $L$ tokens, and each node is a destination of at most $L$ tokens. The goal is to route the tokens to their destinations.

\begin{lemma}\label{lem:task1}
Suppose there exists a \congest algorithm that solves \textbf{Task 1'} with $T_{1'}^{\sf{pre}}(n)$ preprocessing time and $T_{1'}^{\sf{pre}}(n, L)$ query time. Then, there exists an algorithm that solves \textbf{Task 1} with $T_{\rm hie}(n) + T_{\rm sort}^{\sf{pre}}(|W|) + T_{1'}^{\sf{pre}}(n)$ preprocessing time and $O(T_{\rm sort}(|W|, L) + T_{1'}(n, L)) + \poly(\psi^{-1}, \log^{1/\epsilon} n)$ query time, where $T_{\rm hie}(n)$ is the time to build the hierarchical decomposition $\oldmathcal{T}$, and $W$ is the root of $\oldmathcal{T}$.
\end{lemma}
\begin{proof}%
During the preprocessing phase, we first build the hierarchical decomposition. Then, we invoke the preprocessing algorithms for expander sorting and {\bf Task 1'}. Given a query, it suffices to show that all vertices and all tokens can translate their IDs and destination IDs to $\{1, 2, \ldots, n\}$ properly.  The algorithm begins with setting the key $k_z$ of each token $z$ to be its destination ID. The algorithm also creates a dummy token for each vertex in $V$ with the key being its ID. Then, after invoking \Cref{thm:token-ranking}, which takes $O(T_{\rm sort}(|W|, L)) +\poly(\psi^{-1}, \log^{1/\epsilon} n)$ time, the translation is done. Then we can run the query algorithm for {\bf Task 1'}.
\end{proof}

Now we are ready to reduce {\bf Task 1'} to {\bf Task 2}. 
\begin{lemma}\label{lem:task2}
Given a hierarchical decomposition $\oldmathcal{T}$, suppose there exists a \congest algorithm that solves \textbf{Task 2} with $T_{2}^{\sf{pre}}(|X|)$ preprocessing time and $T_{2}^{\sf{pre}}(|X|, L)$ query time for $X \in \oldmathcal{T}$. Then, there exists an algorithm that solves \textbf{Task 1'} with $\poly(\log^{1/\epsilon} n, \psi^{-1}) + T_{2}^{\sf{pre}}(|W|)+T_{2}(|W|, O(1))$ preprocessing time and $T_{2}(|X|, L) +L\cdot \poly(\log^{1/\epsilon} n, \psi^{-1})$ query time, where $W$ is the root of $\oldmathcal{T}$.
\end{lemma}

\begin{proof}%
\item
\paragraph{Preprocessing} 
For each component $X\in\T$, the algorithm obtains the number of best vertices $|X_{best}|$ using a bottom-up approach, which takes $O(D(|H_{X}|\cdot f^{0}_{X}(H_{X})) = \poly(\log^{1/\epsilon} n, \psi^{-1})$ time. Then $W_{best}$ can be propagated to every node in $V$ by using the matching embedding $f_{M_{root}}$ in $Q(f_{M_{root}}) = \poly(\log^{1/\epsilon} n, \psi^{-1})$ time.

Now the algorithm performs preprocessing steps required for {\bf Task 2} in $T_{2}^{\sf{pre}}(|W|)$ time. Then, the algorithm invokes Task 2 on the instance where each vertex $v$ generates a token $z$ with its destination marker $i_z := h(v) = \mathrm{ID}(v)\bmod |V_{best}|$. This can be done in $T_{2}(|W|, O(1))$ time by first throwing the tokens outside of $W$ to $W$ using the matching embedding $f_{M_{root}}$ in $Q(f_{M_{root}}) = \poly(\log^{1/\epsilon} n, \psi^{-1})$ time.

After obtaining the routes for each vertex $v$ to the corresponding best vertex $h(v)$, the routes are memorized in order to serve the query.

\item 
\paragraph{Query} 
For each token $z$ with destination $\mathit{dst}_z$, the algorithm assigns the destination mark $j_z := \mathit{dst}_z \bmod |V_{best}|$. Now the algorithm sends  the tokens outside of $W$ to $W$ using the matching embedding $f_{M_{root}}$ in $L\cdot Q(f_{M_{root}}) = L\cdot \poly(\log^{1/\epsilon} n, \psi^{-1})$ time. Then, the algorithm invokes {\bf Task 2} on $W$ with time $T_2(|W|, L)$.  Finally, the tokens on the best vertices are routed to the destination via the precomputed all-to-best route.
\end{proof}

\begin{theorem}\label{thm:task1}
Fix $\epsilon = \Omega(\sqrt{\log\log n/\log n})$. There exists a \congest algorithm solving \textbf{Task 1} that uses $L\cdot \poly(\psi^{-1})\cdot n^{O(\epsilon)}$ rounds for preprocessing and $L\cdot \poly(\psi^{-1})\cdot \log^{O(1/\epsilon)} n$ rounds for routing.
\end{theorem}
\begin{proof}
By \Cref{lem:task1} and \Cref{lem:task2}, the preprocessing time needed to prepare an instance for {\bf Task 1} is:
$$T_{\rm hie}(n) + T_{\rm sort}^{\sf{pre}}(|W|) + \poly(\log^{1/\epsilon} n, \psi^{-1}) + T_{2}^{\sf{pre}}(|W|)+T_{2}(|W|, O(1))$$
By \Cref{thm:build-hierarchy}, \Cref{thm:expander_sorting_time}, and \Cref{thm:task2}, this quantity is:
$$\poly(\psi^{-1}, n^{\epsilon}, \log^{1/\epsilon} n) = \poly(\psi^{-1}, n^{\epsilon})$$
when $\epsilon =  \Omega(\sqrt{\log\log n/\log n})$.
By \Cref{lem:task1} and \Cref{lem:task2}, the query time is:
$$O(T_{\rm sort}(|W|, L) + T_{2}(|X|, L) +L\cdot \poly(\log^{1/\epsilon} n, \psi^{-1}))$$
By \Cref{thm:expander_sorting_time} and \Cref{thm:task2}, this quantity is:
\[L \cdot \poly(\psi^{-1}, \log^{1/\epsilon} n).\qedhere\]
\end{proof}
\section{Routing on General Expanders}\label{sec:general_graphs}

\paragraph{Expander Split} We follow the framework of \cite{ChangS20}.
A constant degree expander $G^\diamond$ is constructed by replacing each vertex $v$ on $G$ into an expander of $\deg(v)$ vertices $(v, 1), \ldots, (v, \deg(v))$. Then, for each edge incident on $v$ we assign the incident vertex to an arbitrary but unique vertex $(v, i)$.

The main challenge here is that in the actual routing task, the vertices holding the tokens with the same destination vertex $v$ do not have the knowledge to re-assign vertex labels $(v, i)$ in a load-balanced way.
In \cite{ChangS20}, they modified the routing algorithm such that, the destination is represented as a range and the notion of \emph{average load} is introduced.
The modification adds another complication to the algorithm we have developed in \Cref{sec:reductions}.

\paragraph{Local Propagation and Local Serialization to the Rescue} We provide an alternative and conceptually simpler solution of re-assigning destination labels\footnote{The actual implementation may not be simpler, but the correctness should be much simpler to validate.}.
The key to simplifying this reduction is through local propagation (\Cref{lemma:local-propagation}) and local serialization (\Cref{lemma:local-serialization}).
\begin{itemize}[itemsep=0pt]
    \item Local Propagation: The goal is to make sure every token $z$ with destination vertex $\mathit{dst}_z=v$ obtains the knowledge of $\deg(v)$. To achieve this, the algorithm generates a dummy token on each vertex $v\in V$, sets up its grouping key to be $v$, its tag to be $-\infty$, and its variable to be $\deg(v)$. For all query tokens, the algorithm sets the key to be its destination vertex and an arbitrary positive tag.
    Then, after invoking \Cref{lemma:local-propagation}, every token learns $\deg(v)$.
    \item Local Serialization: Each token $z$ obtains a unique serial number $\mathit{SID}_z$ among all tokens with the same destination vertex.
\end{itemize}
Once each token $z$ obtains the information about its destination $\deg(v)$ and $\mathit{SID}_z$, a load-balanced assignment of destination labels can be simply computed by $(v, i:=\mathit{SID}_z\bmod \deg(v)+1)$.

\paragraph{Remark: Unknown Maximum Load} If $L$ is unknown to a routing instance, we can use a standard doubling trick: Set $L'=1, 2, 4, \ldots$ and invoke an expander routing instance with $L=L'$. If at any moment the algorithm realizes some part holds more than the maximum allowed number of tokens, then the current expander routing instance is halted. This doubling trick posed a constant factor to the running time, compared with the case where the maximum load is known.

\section{Expander Routing and Expander Sorting are Equivalent}\label{sec:equivalence}

Let $G$ be an expander graph in \congest model with conductance $\phi$.
We define the following problems on $G$:

\begin{mdframed}
\textsc{ExpanderRouting}$(\mathit{ID}, Z, \mathit{dst}, L)$

\noindent\textbf{Input:}
Each vertex $v$ has a unique identifier $\mathit{ID}(v)\in [1, \poly(n)]$ and holds at most $L$ tokens. 
$Z$ is the set of tokens.
Each token $z\in Z$ has a destination ID $\mathit{dst}_z$ and there are at most $L$ tokens having the same destination ID.

\noindent\textbf{Goal:}
For each $v\in V$, all tokens with destination ID $\mathit{dst}_z=\mathit{ID}(v)$ are located at $v$.
\end{mdframed}

\begin{mdframed}
\textsc{ExpanderSorting}$(\mathit{ID}, Z, \mathit{k}, L)$

\noindent\textbf{Input:}
Each vertex $v$ has an $\mathit{ID}(v)\in [1, \poly(n)]$
and holds at most $L$ tokens.
$Z$ is the set of tokens, and 
each token $z\in Z$ has a key $\mathit{k}_z$.

\noindent\textbf{Goal:}
There are at most $L$ tokens on each vertex.
Furthermore, for each pair of tokens $(z, z')$ on vertices $v$ and $v'$ respectively, $\mathit{ID}(v) \le \mathit{ID}(v')$ implies $k_z\le k_{z'}$.
\end{mdframed}

\begin{lemma}
Suppose there is a \congest algorithm $\oldmathcal{A}_{\textsf{route}}$ that solves \textsc{ExpanderRouting} in $T_{\textsf{route}}(n, \phi, L)$ time.
Then, there is a \congest algorithm $\oldmathcal{A}_{\textsf{sort}}$ that solves \textsc{ExpanderSorting} in $O(\log n)\cdot T_{\textsf{route}}(n, \phi, L)$ time.
\end{lemma}

\begin{proof}
Consider an instance of \textsc{ExpanderSorting}$(\mathit{ID}, Z, k_z, L)$.

Suppose that each vertex $v$ obtains $\mathit{rank}(v)\in \{1, \ldots, n\}$, its rank among all identifiers $\{\mathit{ID}\}$.
To sort the query tokens $Z$, the algorithm first adds dummy tokens (with arbitrary keys) on all vertices such that each vertex has exactly $L$ tokens.
Let $Z^*$ be this extended set of tokens.
Then, following the approach of Su and Vu \cite{su2019distributed},
the algorithm simulates sorting over
an AKS sorting network~\cite{AKS83}.
The sorting network consists of $O(\log n)$ layers, where each layer describes a matching denoting which two sets of tokens are being compared.

Let $M$ be such a matching.
The algorithm calls $\oldmathcal{A}_{\textsf{route}}$ on the following routing task
\textsc{ExpanderRouting}$(\mathit{rank}, Z^*, \mathit{dst}, L)$ such that, 
if a token $z$ is currently on a vertex $u$ that is currently a matched edge $(u, v)\in M$ with $\mathit{rank}(u)<\mathit{rank}(v)$, then $\mathit{dst}_z \gets \mathit{rank}(v)$. Otherwise, we simply set $\mathit{dst}_z\gets u$ so the token would not move.
After the routing, for each matched edge $(u, v)\in M$, all $2L$ tokens are united on the same vertex. Then, after sorting the tokens locally, the $L$ tokens with smaller keys are sent back to vertex $u$, following the previous route.

The sorting is accomplished after simulating the sorting network.
Now, it suffices to obtain  $\mathit{rank}(v)$ for each $v$.
It turns out we can obtain $\mathit{rank}(v)$ by imposing another expander sorting with customized identifiers
\textsc{ExpanderSorting}$(\mathit{ID}', Z', k, 1)$:
The new IDs are an arbitrary serialization among $\{1, 2, \ldots, n\}$ (which can be done within diameter time $D(G)=O(\phi^{-1}\log n)$ by obtaining in-order labels from an arbitrary BFS tree).
Each vertex $v$ generates one token with the key $k_v := \mathit{ID}(v)$.
After the sorting, each token that is originated from vertex $v$ arrives will be located at a vertex $u$ where  $\mathit{ID}'(u)=\mathit{rank}(v)$.
By reversing the route of the tokens, all vertices now obtain their rank.

From the above discussion, we know that the AKS sorting network is simulated at most twice for the given instance of \textsc{ExpanderSorting}.
Hence, the round complexity of $\oldmathcal{A}_{\textsf{sort}}$ is
\[
T_{\textsf{sort}}(n, \phi, L) = O(\phi^{-1}\log n) + O(\log n)\cdot T_{\textsf{route}}(n, \phi, L).\qedhere
\]
\end{proof}

\begin{lemma}\label{lemma:reduction-from-route-to-sort}
Suppose there is a comparison-based \congest algorithm $\oldmathcal{A}_{\textsf{sort}}$ that solves \textsc{ExpanderSorting} in $T_{\textsf{sort}}(n, \phi, L)$ time.
Then, there is a \congest algorithm $\oldmathcal{A}_{\textsf{route}}$ that solves \textsc{ExpanderRouting} in $O(1)\cdot T_{\textsf{sort}}(n, \phi, 2L)$ time.
\end{lemma}

\begin{proof}
Consider an instance of \textsc{ExpanderRouting}$(\mathit{ID}, Z, \mathit{dst}, L)$.

We follow the meet-in-the-middle recipe described in \Cref{sec:task2-leaf-components}.
The intuition is to first count the number of tokens $N_v$ being routed to $v$ via local aggregation (\Cref{lemma:local-aggregation}), and let each vertex $v$ generate $N_v$ dummy tokens.
Next, by a local serialization (\Cref{lemma:local-serialization}), we can assign odd numbers to real tokens and even numbers to dummy tokens. Real tokens and dummy tokens will then be interleaved on any sorted sequence.
By invoking $\oldmathcal{A}_{\textsf{sort}}$ again with maximum load $2L$ (extra empty tokens are generated to ensure that each vertex holds exactly $2L$ tokens), the algorithm is able to pair up each real token with a dummy token.
Finally, this meet-in-the-middle trick lets each dummy token bring \emph{exactly one} real token to the destination.

In conclusion, our algorithm $\oldmathcal{A}_{\textsf{route}}$ has a round complexity $O(1)\cdot T_{\textsf{sort}}(n, \phi, 2L)$.
\end{proof}

If the expander sorting algorithm $\oldmathcal{A}_{\textsf{sort}}$ is not guaranteed to be comparison-based, the following lemma shows that we can still obtain a reduction with an additional $O(\log n)$ factor.

\begin{lemma}
Suppose there is a \congest algorithm $\oldmathcal{A}_{\textsf{sort}}$ that solves \textsc{ExpanderSorting} in $T_{\textsf{sort}}(n, \phi, L)$ time.
Then, there is a \congest algorithm $\oldmathcal{A}_{\textsf{route}}$ that solves \textsc{ExpanderRouting} in $O(\log_{L+1} n)\cdot T_{\textsf{sort}}(n, \phi, 2L)$ time.
\end{lemma}

\begin{proof}
It suffices to show that the deduplication step in the token ranking procedure (\Cref{thm:token-ranking}) can be implemented using $O(\log_{L+1} n)$ calls to $\oldmathcal{A}_{\textsf{sort}}$ when expander sorting is used as a black-box.
Consider the following algorithm: we repeatedly call $\oldmathcal{A}_{\textsf{sort}}$ for $O(\log_{L+1} n)$ times, with the maximum load being set to be at least 2.
After each execution of the expander sort,
each vertex locally marks all but one token with the same key as duplicate.
It is straightforward to check that, after $O(\log_{L+1} n)$ calls, for each distinct key, there will be at most 2 tokens that are left as not marked. These two tokens locate on different vertices.
By generating a dummy token with key being $-\infty$ and invoking another expander sort, the deduplication step can be accomplished as now these 2 tokens are guaranteed to be on the same vertex.
Once we have the deduplication step, we can implement local propagation (\Cref{lemma:local-propagation}) by reverting the above procedure and propagating the desired information.

Thus, the procedure described in the proof of \Cref{lemma:reduction-from-route-to-sort} can be implemented in $O(\log_{L+1} n)$ calls to $\oldmathcal{A}_{\textsf{sort}}$.
\end{proof}